\documentclass[sigconf]{acmart}

\usepackage{amsthm}
\usepackage{xr}

\usepackage{amsmath}
\usepackage{verbatim}
\usepackage{ifthen}
\usepackage{url}
\usepackage{xspace}
\usepackage{subcaption}
\usepackage{calrsfs}
\usepackage{pifont}
\usepackage{amsfonts}
\usepackage{caption}
\usepackage{booktabs}
\usepackage{multirow}

\usepackage{algpseudocode}
\usepackage[linesnumbered, ruled, vlined]{algorithm2e}

\usepackage{tikz}
\usetikzlibrary{shadows,shadings,shapes.symbols,positioning}

\tikzset{
diagonal fill/.style 2 args={fill=#2, path picture={
\fill[#1, sharp corners] (path picture bounding box.south west) -|
                         (path picture bounding box.north east) -- cycle;}},
reversed diagonal fill/.style 2 args={fill=#2, path picture={
\fill[#1, sharp corners] (path picture bounding box.north west) |- 
                         (path picture bounding box.south east) -- cycle;}}
}

\newcommand{\squishlist}{
 \begin{list}{$\bullet$}
  {  \setlength{\itemsep}{0pt}
     \setlength{\parsep}{3pt}
     \setlength{\topsep}{3pt}
     \setlength{\partopsep}{0pt}
     \setlength{\leftmargin}{2em}
     \setlength{\labelwidth}{1.5em}
     \setlength{\labelsep}{0.5em}
} }
\newcommand{\squishlisttight}{
 \begin{list}{$\bullet$}
  { \setlength{\itemsep}{0pt}
    \setlength{\parsep}{0pt}
    \setlength{\topsep}{0pt}
    \setlength{\partopsep}{0pt}
    \setlength{\leftmargin}{2em}
    \setlength{\labelwidth}{1.5em}
    \setlength{\labelsep}{0.5em}
} }

\newcommand{\squishdesc}{
 \begin{list}{}
  {  \setlength{\itemsep}{0pt}
     \setlength{\parsep}{3pt}
     \setlength{\topsep}{3pt}
     \setlength{\partopsep}{0pt}
     \setlength{\leftmargin}{1em}
     \setlength{\labelwidth}{1.5em}
     \setlength{\labelsep}{0.5em}
} }

\newcommand{\squishend}{
  \end{list}
}

\newcommand{\spara}[1]{\smallskip\noindent{\bf #1}}
\newcommand{\mpara}[1]{\medskip\noindent{\bf #1}}

\newcommand{\fpr}[1]{\mathopen{}\left(#1\right)}

\newcommand{\dispfunc}[2]{%
  \ensuremath{%
    \ifthenelse{\equal{\noexpand#2}{}}%
	     {#1}%
		      {{#1}\fpr{#2}}}}

\newcommand{\integers}{\ensuremath{\mathbb{N}}\xspace}

\DeclareMathAlphabet{\pazocal}{OMS}{zplm}{m}{n}

\newcommand{\bigO}{\ensuremath{\mathcal{O}}\xspace}
\newcommand{\np}{\ensuremath{\mathbf{NP}}\xspace}
\newcommand{\p}{\ensuremath{\mathbf{P}}\xspace}

\newcommand{\fpt}{\ensuremath{\textsf{FPT}}\xspace}

\newcommand{\wone}{\ensuremath{\mathbf{W[1]}}\xspace}
\newcommand{\wtwo}{\ensuremath{\mathbf{W[2]}}\xspace}

\newcommand{\seth}{\ensuremath{\textsf{SETH}}\xspace}
\newcommand{\gapeth}{\ensuremath{\textsf{Gap\text{-}ETH}}\xspace}

\newcommand{\note}[1]{{\textcolor{red}{NOTE: #1}}}

\newcommand{\pc}[1]{{\pazocal{#1}}}
\newcommand{\mc}[1]{{\mathcal{#1}}}

\newcommand{\vertexcover}{\ensuremath{\text{\sc Vertex\-Cover}}\xspace}

\newcommand{\matmedian}{\ensuremath{\text{\sc Matroid\-Med\-ian}}\xspace}
\newcommand{\rbmedian}{\ensuremath{rb\text{\sc-Median}}\xspace}

\newcommand{\dominatingset}{\ensuremath{k\text{\sc -Dom\-Set}}\xspace}

\newcommand{\divkmedian}{\ensuremath{\text{\sc Div-}k\text{\sc-Med}}\xspace}
\newcommand{\divkmediannif}{\ensuremath{\text{\sc Div-}k\text{\sc-Median}_{\emptyset}}\xspace}

\newcommand{\kmedianppm}{\ensuremath{k\text{\sc-Med-}p\text{\sc-PM}}\xspace}
\newcommand{\kmedianpm}{\ensuremath{k\text{\sc-Med-}k\text{\sc-PM}}\xspace}
\newcommand{\kmedian}{\ensuremath{k\text{\sc{-Median}}}\xspace}

\newcommand{\kmeans}{\ensuremath{k\text{\sc{-Means}}}\xspace}
\newcommand{\divkmeans}{\ensuremath{\text{\sc Div-}k\text{\sc-Means}}\xspace}
\newcommand{\reqsat}{\ensuremath{{\vec{r}}\text{\sc{-sat}}}\xspace}

\newcommand{\setmcwr}{{\sc SetMulticover-WR}\xspace}

\newcommand{\lszero}{\ensuremath{\small\text{\sf LS}_0}\xspace}
\newcommand{\lszeropswap}{\ensuremath{\small\text{\sf LS}_0(p)}\xspace}

\newcommand{\lsone}{\ensuremath{\small\text{\sf LS}_1}\xspace}

\newcommand{\KM}{\ensuremath{\small\text{\sf KM}}\xspace}
\newcommand{\DP}{\ensuremath{\small\text{\sf DP}}\xspace}
\newcommand{\LP}{\ensuremath{\small\text{\sf LP}}\xspace}
\newcommand{\ES}{\ensuremath{\small\text{\sf ES}}\xspace}

\newcommand{\cost}{\ensuremath{\text{\sf cost}}\xspace}
\newcommand{\impr}{\ensuremath{\text{\sf improv}}\xspace}
\newcommand{\coreset}{\ensuremath{\text{\sf core-set}}\xspace}

\newcommand{\divkins}{\ensuremath{((U,d),F,C,\mathcal{G},\vec{r},k)}\xspace}
\newcommand{\kmedkpmins}{\ensuremath{((U,d),\{E_1^*,\cdots,E_k^*\},C,k))}\xspace}

\newcommand{\poly}{\textsf{poly}\xspace}
\newcommand{\wrt}{\text{w.r.t}\xspace}

\newcommand{\groups}{\mathcal{G}}

\newcommand{\pattern}{constraint pattern\xspace}
\newcommand{\patternset}[1]{E(#1)\xspace}
\newcommand{\charpart}{\mathcal{E}\xspace}
\newcommand{\charvec}{\vec{\chi}\xspace}

\newcommand{\smallball}{\Pi\xspace}

\newcommand{\req}[1]{{\ensuremath{r[#1]}}\xspace}
\newcommand{\reqvec}{{\ensuremath{\vec{r}}}\xspace}

\newcommand{\tikzscale}{{0.7}}

\usepackage{tikz,pgfplots,pgfplotstable}
\usetikzlibrary{positioning,fit,shapes}
\usetikzlibrary{decorations.pathreplacing}
\usetikzlibrary{decorations.markings}
\usetikzlibrary{arrows}

\pgfdeclarelayer{background}
\pgfdeclarelayer{foreground}
\pgfsetlayers{background,main,foreground}

\makeatletter
\tikzset{multicircle/.style  args={#1, #2}{%
 alias=tmp@name, %
  postaction={%
    insert path={
     \pgfextra{%
     \pgfpointdiff{\pgfpointanchor{\pgf@node@name}{center}}%
                  {\pgfpointanchor{\pgf@node@name}{east}}%
     \pgfmathsetmacro\insiderad{\pgf@x}%
        \fill[white] (\pgf@node@name.center)  circle (\insiderad-\pgflinewidth);%
        \draw[#2] (\pgf@node@name.center)  circle (\insiderad-\pgflinewidth);%
        \fill[#2] (\pgf@node@name.center)  -- ++(0:\insiderad-\pgflinewidth) arc (0:#1:\insiderad-\pgflinewidth)--cycle;%
        }}}}}
\makeatother

\definecolor{yafaxiscolor}{rgb}{0.3, 0.3, 0.3}
\definecolor{yafcolor1}{rgb}{0.4, 0.165, 0.553}
\definecolor{yafcolor2}{rgb}{0.949, 0.482, 0.216}
\definecolor{yafcolor3}{rgb}{0.47, 0.549, 0.306}
\definecolor{yafcolor4}{rgb}{0.925, 0.165, 0.224}
\definecolor{yafcolor5}{rgb}{0.141, 0.345, 0.643}
\definecolor{yafcolor6}{rgb}{0.965, 0.933, 0.267}
\definecolor{yafcolor7}{rgb}{0.627, 0.118, 0.165}
\definecolor{yafcolor8}{rgb}{0.878, 0.475, 0.686}
\definecolor{yafcolor9}{rgb}{0.965, 0.733, 0.767}

\newlength{\yafaxispad}
\newlength{\yaftlpad}
\newlength{\yaflabelpad}
\newlength{\yafaxiswidth}
\newlength{\yafticklen}
\makeatletter
\def\pgfplots@drawtickgridlines@INSTALLCLIP@onorientedsurf#1{}
\makeatother

\newcommand{\yafdrawxaxis}[2]{
  \pgfplotstransformcoordinatex{#1}\let\xmincoord=\pgfmathresult 
  \pgfplotstransformcoordinatex{#2}\let\xmaxcoord=\pgfmathresult 
  \pgfsetlinewidth{\yafaxiswidth} 
  \pgfsetcolor{yafaxiscolor}
  \pgfpathmoveto{\pgfpointadd{\pgfpointadd{\pgfplotspointrelaxisxy{0}{0}}{\pgfqpointxy{\xmincoord}{0}}}{\pgfqpoint{-0.5\yafaxiswidth}{\yafaxispad}}}
  \pgfpathlineto{\pgfpointadd{\pgfpointadd{\pgfplotspointrelaxisxy{0}{0}}{\pgfqpointxy{\xmaxcoord}{0}}}{\pgfqpoint{0.5\yafaxiswidth}{\yafaxispad}}}
  \pgfusepath{stroke}

}
\newcommand{\yafdrawyaxis}[2]{
  \pgfplotstransformcoordinatey{#1}\let\ymincoord=\pgfmathresult 
  \pgfplotstransformcoordinatey{#2}\let\ymaxcoord=\pgfmathresult 
  \pgfsetlinewidth{\yafaxiswidth} 
  \pgfsetcolor{yafaxiscolor}
  \pgfpathmoveto{\pgfpointadd{\pgfpointadd{\pgfplotspointrelaxisxy{0}{0}}{\pgfqpointxy{0}{\ymincoord}}}{\pgfqpoint{\yafaxispad}{-0.5\yafaxiswidth}}}
  \pgfpathlineto{\pgfpointadd{\pgfpointadd{\pgfplotspointrelaxisxy{0}{0}}{\pgfqpointxy{0}{\ymaxcoord}}}{\pgfqpoint{\yafaxispad}{0.5\yafaxiswidth}}}
  \pgfusepath{stroke}
}

\pgfplotscreateplotcyclelist{yaf}{%
{yafcolor1,mark options={scale=0.75},mark=o}, 
{yafcolor2,mark options={scale=0.75},mark=square},
{yafcolor3,mark options={scale=0.75},mark=triangle},
{yafcolor4,mark options={scale=0.75},mark=o},
{yafcolor5,mark options={scale=0.75},mark=o},
{yafcolor6,mark options={scale=0.75},mark=o},
{yafcolor7,mark options={scale=0.75},mark=o},
{yafcolor8,mark options={scale=0.75},mark=o}} 

\pgfplotsset{axis y line=left, axis x line=bottom,
  tick align=outside,
  compat = 1.3,
  tickwidth=\yafticklen,
  clip = false,
  every axis title shift = 0pt,
    x axis line style= {-, line width = 0pt, opacity = 0},
    y axis line style= {-, line width = 0pt, opacity = 0},
    x tick style= {line width = \yafaxiswidth, color=yafaxiscolor, yshift = \yafaxispad},
    y tick style= {line width = \yafaxiswidth, color=yafaxiscolor, xshift = \yafaxispad},
    x tick label style = {font=\scriptsize, yshift = \yaftlpad},
    y tick label style = {font=\scriptsize, xshift = \yaftlpad},
    every axis y label/.style = {at = {(ticklabel cs:0.5)}, rotate=90, anchor=center, font=\scriptsize, yshift = -\yaflabelpad},
    every axis x label/.style = {at = {(ticklabel cs:0.5)}, anchor=center, font=\scriptsize, yshift = \yaflabelpad},
    x tick label style = {font=\scriptsize, yshift = 1pt},
    grid = major,
    major grid style  = {dash pattern = on 1pt off 3 pt},
  every axis plot post/.append style= {line width=\yafaxiswidth} ,
  legend cell align = left,
  legend style = {inner sep = 1pt, cells = {font=\scriptsize}},
  legend image code/.code={%
    \draw[mark repeat=2,mark phase=2,#1] 
    plot coordinates { (0cm,0cm) (0.15cm,0cm) (0.3cm,0cm) };%
  } 
}

\tikzset{
  on each segment/.style={
    decorate,
    decoration={
      show path construction,
      moveto code={},
      lineto code={
        \path [#1]
        (\tikzinputsegmentfirst) -- (\tikzinputsegmentlast);
      },
      curveto code={
        \path [#1] (\tikzinputsegmentfirst)
        .. controls
        (\tikzinputsegmentsupporta) and (\tikzinputsegmentsupportb)
        ..
        (\tikzinputsegmentlast);
      },
      closepath code={
        \path [#1]
        (\tikzinputsegmentfirst) -- (\tikzinputsegmentlast);
      },
    },
  },
  mid arrow/.style={postaction={decorate,decoration={
        markings,
        mark=at position .6 with {\arrow[#1]{stealth}}
      }}},
}

\setcopyright{none}
\renewcommand\footnotetextcopyrightpermission[1]{}

\begin{document}

\title{Clustering with fair-center representation:
parameterized approximation algorithms and heuristics}

\author{Suhas Thejaswi}
\affiliation{%
  \institution{Aalto University}
  \city{}
  \country{Finland}}
\email{firstname.lastname@aalto.fi}
\author{Ameet Gadekar}
\affiliation{%
  \institution{Aalto University}
  \city{}
  \country{Finland}}
\email{firstname.lastname@aalto.fi}
\author{Bruno Ordozgoiti}
\affiliation{%
  \institution{Queen Mary University of London}
  \city{}
  \country{United Kingdom}}
\email{b.ordozgoiti@qmul.ac.uk}
\author{Michal Osadnik}
\affiliation{%
  \institution{Aalto University}
  \city{}
  \country{Finland}}
\email{firstname.lastname@aalto.fi}

\begin{abstract}
We study a variant of classical clustering formulations in the context of
algorithmic fairness, known as diversity-aware clustering. In this variant we
are given a collection of facility subsets, and a solution must contain at least
a specified number of facilities from each subset while simultaneously
minimizing the clustering objective ($k$-median or $k$-means).
We investigate the fixed-parameter tractability of these
problems and show several negative hardness and inapproximability results, even
when we afford exponential running time with respect to some parameters.

Motivated by these results we identify natural parameters of the problem, and present fixed-parameter approximation algorithms with approximation ratios 
$\big(1 + \frac{2}{e} +\epsilon \big)$ and 
$\big(1 + \frac{8}{e}+ \epsilon \big)$ for \textit{diversity-aware $k$-median} and
\textit{diversity-aware $k$-means} respectively, and argue that these ratios are essentially tight assuming
the gap-exponential time hypothesis. We also present a simple and more practical
bicriteria approximation algorithm with  better running time
bounds. We finally propose efficient and practical heuristics. We
evaluate the scalability and effectiveness of our methods in a wide
variety of rigorously conducted experiments, on both real and synthetic
data. 

\end{abstract}

\maketitle


\begin{acks}
This work is supported by Academy of Finland projects AIDA (317085) and MLDB
(325117), ERC grant under the EU Horizon 2020 research and innovation programme
(759557), Nokia foundation scholarship (20220290). Part of this work was done
while Bruno Ordozgoiti was a postdoctoral researcher at Aalto University.
\end{acks}

\section{Introduction}
Consider the problem of forming a representative committee. In
essence, the task amounts to finding a group of people among a given
set of candidates, to represent the will of a (usually) larger body of
constituents. In computational terms, this can be modeled as a
clustering problem like \kmedian: each cluster center is a chosen
candidate, and the points in the corresponding cluster are the
constituents it best represents. 

In certain scenarios, it may be adequate to consider additional requirements.
For instance, it may be necessary that at least a number of the chosen committee
members belong to a certain minority-ethnic background, to ensure that all
groups in a society are represented in the decision-making process. 
This problem was recently formalized as the \textit{diversity-aware $k$-median}
problem (\divkmedian)~\cite{thejaswi2021diversity}. As in conventional \kmedian,
the goal is to pick $k$ facilities to minimize the sum of distances from clients
to their closest facility~\cite{arya2001local}. In \divkmedian, however, each
facility is associated to an arbitrary number of attributes from a given finite
set. The solution is required to contain at least a certain number of facilities
having each attribute (the requirement for each attribute is given to us as part
of the input).

Thejaswi et al. showed that this additional constraint makes \kmedian harder to
solve~\cite{thejaswi2021diversity}, in the following sense. While \kmedian is
\np-hard to solve exactly, it is \np-complete to even decide whether a
\divkmedian instance has a feasible solution. The rest of their work, thus,
focuses on tractable cases and practical heuristics.
This work follows a recent line of interest in \textit{algorithmic fairness},
which has attracted significant attention in recent years. In the design of fair
algorithms, additional constraints are imposed on the objective function to
ensure equitable ---or otherwise desirable--- outcomes for the different groups
present in the
data~\cite{zafar2017fairness,dwork2018decoupled,chierichetti2017fair,schmidt2019fair,huang2019coresets,backurs2019scalable,bercea2019cost}.

\mpara{Contributions.} In this paper we provide a much more comprehensive analysis of
\divkmedian (\divkmeans resp.), addressing computational complexity,
approximation algorithms, and practical heuristics for the problems. In particular, we
give the first known and tight approximation results for the problems.

Since we know that \divkmedian does not admit polynomial-time
approximation algorithms~\cite{thejaswi2021diversity},
we first focus on \textit{fixed-parameter
  tractability}~\cite{cygan2015parameterized}. That is, we seek to
answer the following question: can we hope to find
algorithms with approximation guarantees by allowing their 
running time to be exponential in some input parameter? In other
words, is approximating \divkmedian (\divkmeans resp.) fixed-parameter tractable (\fpt)?

Our main
result in this paper is a positive answer to this question. We give a constant-factor
approximation algorithm with running time parameterized by $k$ and $t$, the
number of clusters and the number of candidate attributes
respectively. We further develop our understanding of \divkmedian (\divkmeans resp.) by
characterizing the problem in terms of parameterized complexity. Finally, we consider practical aspects of the problem
and design effective, practical heuristics which we evaluate through a
variety of experiments.
Our contributions are summarized below:

\spara{Computational complexity.} We strengthen the known complexity results for
\divkmedian (\divkmeans resp.) by analyzing its parameterized complexity and
inapproximibility. In particular, for these problems,
\squishlisttight
\item we give a lower bound for the running time of optimal and approximation
algorithms (Corollary~\ref{corollary:introduction:1});
\item we show that finding bicriteria approximation algorithms  is
fixed-parameter intractable with respect to the number of clusters
(Proposition~\ref{proposition:introduction:2});
\item we show that they are fixed-parameter intractable with
respect to various choices of parameters (Proposition~\ref{proposition:introduction:3}).
\squishend
\spara{Approximation algorithms.}
\squishlisttight
\item We give the first and tight fixed-parameter
tractable constant-factor approximation algorithm for
\divkmedian (\divkmeans resp.) w.r.t. $k$ and $t$ (Theorem~\ref{theorem:mainfptapx}).
\item We give a faster and tight dynamic programming algorithm for deciding the
feasibility (Theorem~\ref{thm:dpfeasibility}). This yields a faster bicriteria
approximation algorithm (Theorem~\ref{theorem:localsearch}).
\squishend
\spara{Practical heuristics and empirical evaluation.}
\squishlisttight
\item Despite their theoretical guarantees, the methods discussed above are
impractical. We propose a practical approach to find feasible solutions based on
linear programming. Despite its lack of guarantees, we show how it can be used
as a building block in the design of effective heuristics.
\item We evaluate the proposed methods in a wide variety of experimental
results, rigorously conducted on real an synthetic datasets.  In particular, we
show that the proposed heuristics are able to reliably and efficiently
find feasible solutions of good quality on a variety of real data sets.
\squishend

The rest of the paper is organized as follows. In
Section \ref{sec:preliminaries} we introduce notation and basic
notions. In Section \ref{sec:hardness} we present our
 computational complexity analysis, and Section \ref{sec:results} gives an overview of our main results. Our
algorithms are described in Sections \ref{sec:algorithm} and
~\ref{sec:bicri}, and our experimental results in Section
\ref{sec:experiments}. An overview of related work is given in Section
\ref{sec:related}, while Section \ref{sec:conclusions} provides
concluding remarks.
Some proofs are deferred to the
Supplementary.

\section{Preliminaries}
\label{sec:preliminaries}
In this section we introduce notation and problem definitions.

\mpara{Notation.}
\label{sec:preliminaries:notation}
Given a metric space $(U,d)$, a set $C \subseteq U$ of clients, a set $F
\subseteq U$ of facilities and a subset $S \subseteq F$ of facilities, we denote
by $\cost(S)=\sum_{c \in C} d(c,S)$ the clustering cost of $S$, where
$d(c,S)=\min_{s \in S} d(c,s)$. 
We say that $C$ is weighted when every $c \in C$ is associated to a weight $w_c
\in \mathbb{R}$, and the clustering cost becomes
$\cost(S)= \sum_{c \in C} w_c \cdot d(c,S)$.
Similarly, for $C' \subseteq C$ and $S \subseteq F$, we write 
$\cost(C',S)=\sum_{c \in C'} w_c \cdot d(c,S)$.  
Further, given a collection $\mathcal{G}=\{G_1,\dots,G_t\}$ of facility groups
such that $G_i \subseteq F$, for each facility $f \in F$ we denote by
$\vec{\chi}_f \in \{0,1\}^t$ the \emph {characteristic vector} of $f$ with
respect to $\mathcal{G}$, and is defined as
$\vec{\chi}_f[i] = 1$ if $f \in G_i$, $0$ otherwise, for all $i \in [t]$. For
$\eta >0$ and $a \in \mathbb{Z}_{\ge 0}$, we denote $[a]_\eta \in \mathbb{Z}$ as
the smallest integer such that $(1+\eta)^{[a]_\eta} \ge a$.
For a metric space $(U,d)$, the aspect ratio is defined as $\Delta =
\frac{\max_{x,y\in U} d(x,y)}{\min_{x,y\in U} d(x,y)}$.

In this paper we use standard parameterized complexity terminology from 
Cygan et al.~\cite{cygan2015parameterized}.


\begin{definition}[\bf Diversity-aware $k$-median (\divkmedian)]
\label{def:divkmedian}
Given a metric space $(U,d)$, a set $C \subseteq U$ of clients, 
a set $F \subseteq U$ of facilities,
a collection, called \textbf{groups}, $\mathcal{G}=\{G_1,\dots,G_t\}$
of facility sets $G_i
\subseteq F$, a budget $k \leq |F|$, and a vector of
requirements $\vec{r}=(r[1],\dots,r[t])$.
%
The problem asks to find a subset of facilities $S \subseteq F$ of size $k$,
satisfying $|S \cap G_i| \geq r[i]$ for all $i \in [t]$, such that the clustering
cost of $S$, $\cost(S) = \sum_{c \in C} d(c, S)$ is minimized. An instance of
\divkmedian is denoted as $I=\divkins$.
\end{definition}

In {\bf diversity-aware $k$-means} problem (\divkmeans), the clustering cost of
$S \subseteq F$ is $\cost(S) = \sum_{c \in C} d(c,S)^2$.
We denote $r= \max_{i \in [t]} r[i]$ and we assume $\Delta$ is polynomially bounded \wrt $|U|$~\cite{cohen2019tight}.

\section{Hardness}
\label{sec:hardness}
To motivate the choice of parameters for designing \fpt algorithms,
we characterize the hardness of \divkmedian (\divkmeans resp.) based on standard
complexity theory assumptions.
Observe that \divkmedian (\divkmeans resp.) problems are an amalgamation of two independent
problems: ($i$) finding a subset of facilities $S \subseteq F$ of size $|S|=k$
that satisfies the requirements $|S \cap G_i| \geq \req{i}$ for all $i \in [t]$, and
($ii$) minimizing the \kmedian (\kmeans resp.) clustering cost.  To remain consistent with the
problem statement of Thejaswi et al.~\cite{thejaswi2021diversity}, we refer to
subproblem ($i$) as the requirement satisfiability problem (\reqsat), where the cost of clustering is
ignored. If we ignore the requirements in ($i$) we obtain the classical \kmedian
(\kmeans resp.)
formulation, which immediately establishes the \np-hardness of \divkmedian (\divkmeans resp.).

A reduction of the vertex cover problem to \reqsat is sufficient to show
that \divkmedian (\divkmeans resp.) are inapproximable to any multiplicative factor in
polynomial-time even if all the subsets are of size
two~\cite[Theorem~3]{thejaswi2021diversity-arxiv}.
The \wtwo-hardness of \divkmedian (\divkmeans resp.) with respect to parameter $k$ is a consequence
of the fact that \kmedian (\kmeans resp.) are \wtwo-hard, which follow from a reduction by Guha
and Kuller~\cite{guha1998greedy}.
More strongly, combining the result of \cite[Lemma~1]{thejaswi2021diversity}
with the strong exponential time hypothesis (\seth)~\cite{impagliazzo2001on}, we conclude the following: if we only consider the parameter $k$, a
trivial exhaustive-search algorithm is our best hope for finding an optimal, or
even an approximate, solution to \divkmedian (\divkmeans resp.). The proof is
in Supplementary~\ref{app:otherproofs}.
\begin{corollary}
\label{corollary:introduction:1}
Assume \seth. For all $k \geq 3$ and $\epsilon >0$, there exists no 
$\bigO(|{F}|^{k - \epsilon})$
algorithm to solve \divkmedian (\divkmeans resp.).
Furthermore, there exists no
$\bigO(|{F}|^{k - \epsilon})$
algorithm to approximate \divkmedian (\divkmeans resp.) to any multiplicative factor.
\end{corollary}

Given the $\wtwo$-hardness of \divkmedian with respect to parameter $k$, it is
natural to consider relaxations of the problem. An obvious question is
whether we can approximate \divkmedian  in \fpt time \wrt $k$, if we are allowed
to open, say, $f(k)$ facilities instead of $k$, for some function $f$.
Unfortunately, this is also unlikely as \divkmedian captures
\dominatingset~\cite[Lemma~1]{thejaswi2021diversity}, and even finding a
dominating set of size $f(k)$ is \wone-hard~\cite{karthik2019on}. The proof is
in Supplementary~\ref{app:otherproofs}.
\begin{proposition}
\label{proposition:introduction:2}
For any function $f(k)$, finding $f(k)$ facilities that approximate the
\divkmedian (\divkmeans resp.) cost to any multiplicative factor in \fpt time with respect to
parameter $k$ is $\wone$-hard.
\end{proposition}

A possible way forward is to identify other parameters of the problem, to design
\fpt algorithms to solve the problem optimally. As established earlier,
\kmedian is a special case of \divkmedian when $t=1$. This immediately rules out
an exact \fpt algorithm for \divkmedian with respect to parameters $(k,t)$.
Furthermore, we caution the reader against entertaining the prospects of other,
arguably natural, parameters, such as the maximum lower bound $r =
\max_{i\in[t]} r[i]$ (ruled out by the relation $r[i]\leq k$) and the maximum
number of groups a facility can belong to 
$\mu =\max_{f \in F}(|{G_i}: f \in G_i, i \in [t]|)$ 
(ruled out by the relaxation $\mu \leq t$).
\begin{proposition}
\label{proposition:introduction:3}
Finding an optimal solution for \divkmedian  (\divkmeans resp.) is $\wtwo$-hard with respect to parameters
$(k,t)$, $(r,t)$ and $(\mu, t)$.
\end{proposition}

The above intractability results thwart our hopes of solving the \divkmedian
problem optimally in \fpt time. We then ask, are there any parameters of the
problem that allow us to find an approximate solution in \fpt time? We answer
this question positively, and present a tight \fpt-approximation algorithm \wrt
$(k,t)$, the number of chosen facilities and the number
of facility groups.

\section{Our results}
\label{sec:results}
Our main result, stated below, shows that a constant-factor approximation of
\divkmedian (\kmeans resp.) can be achieved in \fpt time with respect to parameters $(k,t)$. In
fact, somewhat surprisingly, the factor is the same as the one achievable for
\kmedian (\kmeans resp.). So despite the stark contrast in their polynomial-time
approximability, the \fpt landscape is rather similar for these two problems. We
also note that under the gap-exponential time hypothesis (\gapeth), the
approximation ratio achieved in Theorem~\ref{theorem:mainfptapx} is essentially
tight for any \fpt algorithm \wrt $(k,t)$. This follows from combining the fact
that the case $t=1$ is essentially \kmedian (\kmeans resp.) with the results of Cohen-Addad et
al.~\cite{cohen2019tight}, which assuming \gapeth gives a lower bound for any
\fpt algorithm \wrt $k$. This bound matches their ---and our--- approximation
guarantee.
\begin{theorem} 
\label{theorem:mainfptapx}
For every $\epsilon>0$, there exists a randomized $(1 + \frac{2}{e}
+\epsilon)$-approximation algorithm for \divkmedian with time
time $f(k,t,\epsilon) \cdot \poly(|U|)$, where $f(k,t,\epsilon) =
\bigO\left(\left( \frac{2^t k^3 \log^2 k}{\epsilon^2 \log(1+\epsilon)}\right)^k \right)$. 
Furthermore, the approximation ratio is tight for any \fpt algorithm
\wrt $(k,t)$, assuming \gapeth.
For \divkmeans, with the same running time, we obtain $(1 + \frac{8}{e} +\epsilon)$-approximation, which is tight assuming \gapeth.
\end{theorem}

Finally, in Section~\ref{sec:bicri} we will point out a simple observation: by
relaxing the upper bound on the number of facilities to at most $2k$, we can use
a practical local-search heuristic and obtain a slightly weaker quality
guarantee with better running time bounds.
To achieve this we make use of Theorem~\ref{thm:dpfeasibility}.
\begin{theorem}
\label{theorem:localsearch} 
For every $\epsilon > 0$, there exists a
randomized $(3+\epsilon)$-approximation algorithm that outputs at most $2k$
facilities for the \divkmedian problem in time $\bigO( 2^t(r+1)^t \cdot \poly(|U|,
1/\epsilon))$.
\end{theorem}

\section{Algorithms}
\label{sec:algorithm}
In this section we present an \fpt approximation algorithm for \divkmedian. 
For \divkmeans, the ideas are similar.
Throughout the section, by \fpt we imply \fpt \wrt $(k,t)$.

A birds-eye view of our algorithm (see Algorithm~\ref{algo:divkmed}) is as follows: Given a feasible instance
of \divkmedian, we first carefully enumerate collections of
facility subsets that satisfy the
lower-bound requirements (Section~\ref{sec:algorithm:finding}). 
For each such collection, we obtain a constant-factor approximation of the
optimal cost of the collection (Section~\ref{sec:algoritheorem:fpt}). 
%
Since at least one of these feasible 
solutions is optimal, 
the corresponding approximate solution will be an approximate
solution for the \divkmedian problem. 
A key ingredient for obtaining a constant factor approximation in \fpt
time is to shrink the set of clients. For
this we rely on the notion of coresets (Section ~\ref{sec:algoritheorem:coresets}).


In the exposition to follow, we will refer to the problem of \kmedian with 
$p$-partition matroid constraints:

\begin{definition}[\kmedian with $p$-Partition Matroid (\kmedianppm)]
\label{def:kmedianpm}
Given a metric space $(U,d)$, a set of clients $C \subseteq U$, a set of
facilities $F \subseteq U$ and a collection $\mathcal{E}=\{E_1,\dots,E_p\}$ of
disjoint facility groups called a $p$-{\em partition matroid}.
The problem asks to find a subset of facilities $S \subseteq F$ of size $k$,
containing at most one facility from each group $E_i$, so that the clustering
cost of $S$, $\cost(S)=\sum_{c\in C}d(c,S)$ is minimized. An instance of
\kmedianppm is specified as $J=((U,d),F,C,\mathcal{E},k)$.
\end{definition}

\subsection{Finding feasible {\pattern}s}
\label{sec:algorithm:finding}
We start by defining the concept of {\pattern}.
Given an instance $I=\divkins$ of \divkmedian, where
$\mathcal{G}=\{G_1,\dots,G_t\}$, consider the set  $\{\charvec_f\}_{f\in F}$ of the characteristic vectors of $F$.
For each $\vec{\gamma} \in \{0,1\}^t$, let 
$\patternset{\vec{\gamma}} = \{f \in F: {\charvec}_f=\vec{\gamma}\}$ 
denote the set of all facilities with
characteristic vector $\vec{\gamma}$.  Finally, let $\charpart=\{E(\vec{\gamma}):
\vec{\gamma} \in \{0,1\}^t\}$.
Note that $\charpart$ induces a partition on $F$.

Given a $k$-multiset
$\pazocal{E}=\{E(\vec{\gamma}_{i_1}),\dots,E(\vec{\gamma}_{i_k})\}$, where each $E(\vec{\gamma}_{i_j})
\in \charpart$,
the \textit{\pattern} associated with $\pazocal{E}$ is the vector obtained by the
element-wise sum of the characteristic vectors 
$\{\vec{\gamma}_{i_1},\dots,\vec{\gamma}_{i_k}\}$, that is, 
$\sum_{j \in [k]} \vec{\gamma}_{i_j}$.
A \pattern is said to be \textit{feasible} if 
$\sum_{j \in [k]} \vec{\gamma}_{i_j} \geq \vec r$, 
where the inequality is taken element-wise.

\begin{lemma} 
\label{lemma:feasiblecp}
Given an instance $I=\divkins$ of \divkmedian, we can enumerate all the $k$-multisets with feasible constraint pattern in time $\bigO(2^{tk}t |U|)$. 
\end{lemma}
\begin{proof}
There are $|\charpart|+k-1 \choose k$
possible $k$-multisets of $\charpart$, so enumerating all feasible {\pattern}s
can be done in $\bigO(|\charpart|^k t)$ time.
Further, the enumeration of $\charpart$ itself can be done in time $\bigO(2^t|F|)$, since $|\charpart| \leq 2^t$. Hence,
time complexity of enumerating all feasible {\pattern}s
is $\bigO(2^{tk} t |U|)$.
\end{proof}
Observe that for every $k$ multiset
$\pazocal{E}=\{E(\vec{\gamma}_{i_1}),\dots,E(\vec{\gamma}_{i_k})\}$ with a
feasible  \pattern, picking an arbitrary facility from each
$E(\vec{\gamma}_{i_j})$ yields a feasible solution to the \divkmedian instance $I$.

%
\subsection{Coresets}
\label{sec:algoritheorem:coresets}
Our algorithm relies on the notion of coresets. The high-level idea
is to reduce the number of clients  such that
the distortion in the sum of distances is bounded to a multiplicative factor 
$(1 \pm \nu)$, for
some $\nu > 0$. Given an instance $((U,d),C,F,k)$ of
the \kmedianpm problem, for every $\nu > 0$ we can reduce the number of
clients in $C$ to a weighted set $C'$ of size $|C'|= \bigO({\nu}^{-2}k \log|U|)$.
We make use of the coreset construction for \kmedian by Feldman and
Langberg~\cite{feldman2011unified} and extend the approach to \kmedianpm. To our knowledge, this is the best-known framework for constructing
coresets.
\begin{definition}[Coreset]
Given an instance $I=((U,d),C,F,k)$ of \kmedian and a constant $\nu >0$, a
(strong) {coreset} is a subset of clients $C' \subseteq C$ with associated weights
$\{w_c: c \in C'\}$ such that for any subset of facilities $S \subseteq F$ of size $|S|=k$
it holds that
$$
(1-\nu) \cdot \sum_{c \in C} d(c,S) \le \sum_{c \in C'} w_c \cdot d(c,S) \le
(1+\nu)\cdot \sum_{c \in C} d(c,S)
$$
\end{definition}

\begin{theorem}[\cite{feldman2011unified}, Theorem~4.9]
\label{theorem:coreset}
Given a metric instance  $I=((U,d),C,F,k)$ of the \kmedian problem,
for each $\nu > 0$, $\delta<\frac{1}{2}$,
there exists a randomized algorithm that, with probability at least $1-\delta$,
computes
a coreset $C' \subseteq C$ of
size $|C'| = \bigO({\nu}^{-2}(k \log|U| + \log \frac{1}{\delta}))$ in time
$\bigO(k(|U| + k)+\log^2\frac{1}{\delta}\log^2|U|)$. For \kmeans, with the same runtime, it yields a coreset of size $|C'| = \bigO({\nu}^{-4}(k \log|U| + \log \frac{1}{\delta}))$.
\end{theorem}
Observe that the coresets obtained from the above theorem are also  coresets for \divkmedian and \divkmeans resp., as the corresponding objective remain same.

\subsection{\fpt approximation algorithms}
\label{sec:algoritheorem:fpt}
In this section we present our main result. We will first give an
intuitive overview of our algorithm. As a warm-up, we will describe a
simple $(3+\epsilon)$-\fpt-approximation algorithm. Then, we will show how
to obtain a better guarantee, leveraging the recent \fpt-approximation techniques of \kmedian.
\paragraph*{Intuition}
Given an instance $\divkins$ of \divkmedian, we
first partition the facility set $F$ into at most $2^t$ subsets $\mathcal{E} =
\{E(\vec{\gamma}): \vec{\gamma} \in \{0,1\}^t\}$, such that each
subset $E(\vec{\gamma})$ corresponds to the facilities with
characteristic vector same as $\vec\gamma \in \{0,1\}^t$. Then, using
Lemma~\ref{lemma:feasiblecp}, we enumerate all $k$-multisets of $\mathcal{E}$ with
feasible \pattern. For each such $k$-multiset $\pazocal{E}=\{E(\vec{\gamma}_{1}),\dots,E(\vec{\gamma}_{k})\}$, we generate an instance $J_\pazocal{E}=((U,d), \{E(\vec{\gamma}_{1}),\dots,E(\vec{\gamma}_{k})\},C',k)$ of \kmedianpm,
resulting in at most $\bigO(2^{tk}t|U|)$ instances.
Next, using Theorem~\ref{theorem:coreset}, we build a coreset $C'\subseteq C$ of clients. 
In our final step, we obtain
an approximate solution to each \kmedianpm instance by adapting the techniques
from~\cite{cohen2019tight}, which we discuss next.

Let $\pazocal{E}=\{E(\vec{\gamma}_{1}),\dots,E(\vec{\gamma}_{k})\}$  be a $k$-multiset of $\charpart$ with a feasible constraint pattern,
and let $J_\pazocal{E}$ be the corresponding feasible \kmedianpm
instance. 
Let $\tilde{F}^* = \{\tilde{f}_i^* \in E(\vec{\gamma}_{i})\}_{i \in [k]}$ be an
optimal solution of $J_\pazocal{E}$. For each $\tilde{f}^*_i$, 
let $\tilde{c}^*_i \in C'$ be a closest client, with $d(\tilde{f}^*_i,\tilde{c}^*_i) = \tilde{\lambda}^*_i$. 
Next, for each $\tilde{c}^*_i$ and $\tilde{\lambda}^*_i$, let $\tilde{\Pi}^*_i \subseteq
E(\vec{\gamma}_{i})$ be the set of facilities $f\in
E(\vec{\gamma}_{i})$ such that  $d(f,\tilde{c}^*_i) = \tilde{\lambda}^*_i$. Let us call $\tilde{c}^*_i$ and $\tilde{\lambda}^*_i$ as the leader and radius of $\tilde{\Pi}^*_i$, respectively.
 Observe
that, for each $i \in [k]$,  $\tilde{\Pi}^*_i$ contains  $\tilde{f}^*_i$. Thus, if only we knew $\tilde{c}^*_i$ and
$\tilde{\lambda}^*_i$ for all $i \in [k]$, we would be able to obtain a provably good solution.

To find the closest client $\tilde{c}^*_i$ and its corresponding distance $\tilde{\lambda}^*_i$ in \fpt time, we employ techniques of Cohen-Addad et al.~\cite{cohen2019tight}, which they build on the work of Feldman and Langberg~\cite{feldman2011unified}. 
The idea is to reduce the search spaces small enough so as to allow brute-force search in \fpt time. 
To this end, first, note that, we already have a smaller client set, $|C'| = O(k\nu^{-2}\log |U|)$, since $C'$ is a client coreset.
Hence, to find $\{\tilde{c}^*_i\}_{i \in [k]}$, we enumerate all ordered $k$-multisets of $C'$ resulting in $\bigO((k\nu^{-2}\log |U|)^k)$
time.
Then, to bound the search space of $\lambda^*_i$ (which is at most $\Delta = \poly(|U|)$), we discretize the interval $[1,\Delta]$ to
$[[\Delta]_\eta]$, for some $\eta >0$. Note that $[\Delta]_\eta \le
\lceil \log_{1+\eta} \Delta\rceil = O(\log |U|)$. Hence, enumerating all
ordered $k$-multisets of $[[\Delta]_\eta]$, we spend at most
$\bigO(\log^k |U|)$ time. Thus, the total time for this step, guaranteeing $\tilde{c}^*_i$ and
$\tilde{\lambda}^*_i$ in some enumeration, is $\bigO((k\nu^{-2}\log^2 |U|)^k)$, which is \fpt.

Next, using the facilities in $\{\tilde{\Pi}^*_i\}_{i \in[k]}$, we find an approximate
solution for the instance $J_\pazocal{E}$.
As a warm-up, we show in Lemma~\ref{lemma:threeapx}
that picking exactly one facility from each $\tilde{\Pi}^*_i$ arbitrarily already
gives a $(3+\epsilon)$ approximate solution.
Finally, in Lemma~\ref{lemma:partition}, we obtain a better
approximation ratio by modeling the \kmedianpm problem as a problem of
maximizing a monotone
submodular function, relying on the ideas of Cohen-Addad et
al~\cite{cohen2019tight}.

\begin{lemma} 
\label{lemma:threeapx}
For every $\epsilon>0$, there exists a randomized $(3+\epsilon)$-approximation
algorithm for the \divkmedian problem which runs in time $f(k,t,\epsilon) \cdot
\poly(|U|)$, where
$f(k,t,\epsilon)=\bigO((2^t \epsilon^{-2} k^2 \log k)^k)$.
\end{lemma}
\begin{proof} Let $I=\divkins$ be an instance of \divkmedian.
Let  $J=((U,d),C, \{E_1^*,\dots,E_k^*\},k)$ be an instance of \kmedianpm corresponding
to an optimal solution of $I$. That is, for some optimal solution $F^*=\{f_1^*, \dots,
f_k^*\}$ of $I$, we have $f_j^*\in E_j^*$. Let $c^*_j \in C'$ be the closest client to $f^*_j$, for $j \in [k]$, with $d(f^*_j,c^*_j) = \lambda_j$.
Now, consider the enumeration iteration where leader set is $\{c^*_j\}_{j \in [k]}$ and the radii is $\{\lambda_j^*\}$. 
The construction is illustrated in
Figure~\ref{fig:mainfptapx}.

We define $\Pi^*_i$ to be the set of facilities in $E(\vec{\gamma}_i^*)$
at a distance of at most $\lambda_i^*$ from $c_i^*$. 
We will now argue that picking one arbitrary facility from each $\Pi^*_i$ gives a
$3$-approximation with respect to an optimal pick.
%
Let $C^*_j \subseteq C'$ be a set of clients assigned to each facility $f^*_j$
in optimal solution. Let $\{f_1,\dots,f_k\}$ be the  arbitrarily chosen
facilities, such that $f_j \in \Pi^*_j$. Then for any $c \in C_j$
\[
d(c,f_j) \leq d(c,f_j^*) + d(f_j^*,c_j^*) + d(c_j^*,f_j)).
\]

By the choice of $c_j^*$ we have 
$d(f_j^*,c_j^*) + d(c_j^*,f_j)) \leq 2 \lambda_j^* \leq 2 d(c,f_j^*)$, 
which implies
$\sum_{c \in C_j} d(c,f_j) \leq 3 \sum_{c \in C_j} d(c,f_j^*).$
By the properties of the coreset and bounded discretization
error~\cite{cohen2019tight}, we obtain the approximation stated in the lemma.
\end{proof}

\begin{figure}
\centering
\scalebox{0.6}{\begin{tikzpicture}[scale=\tikzscale,every node/.style={scale=\tikzscale}]

\input{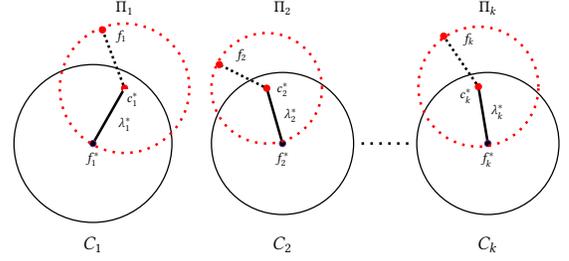}


\node[fcircle, minimum width=5cm] (c1) at (0,0) {$$};
\node[point1] (f1o) at (0,0) {$$};
\node[fill=white] at (0,-0.5) {\LARGE \bf $f_1^*$};
\node[fill=white] at (0,-3.2) {\Huge \bf $C_1$};

\node[fill=white] at (1,4.3) {\huge \bf $\smallball_1$};
\node[dcircle, minimum width=4.1cm] (c1) at (1,1.75) {$$};
\node[point2] (c1) at (1,1.75) {$$};
\node[fill=white] at (1.25,1.4) {\LARGE \bf $c_1^*$};
\node[point2] (f1) at (0.3,3.6) {$$};
\node[fill=white] at (0.9,3.4) {\LARGE \bf $f_1$};

\draw (f1o) edge[exedge] (c1);
\node[fill=white] at (1,0.6) {\LARGE \bf $\lambda_1^*$};
\draw (c1) edge[exedge, dotted] (f1);
\node[fcircle, minimum width=4.5cm] (c2) at (6,0) {$$};
\node[point1] (f2o) at (6,0) {$$};
\node[fill=white] at (6,-0.5) {\LARGE \bf $f_2^*$};
\node[fill=white] at (6,-3.2) {\Huge \bf $C_2$};

\node[fill=white] at (6,4.3) {\huge \bf $\smallball_2$};
\node[dcircle, minimum width=3.5cm] (c1) at (5.5,1.75) {$$};
\node[point2] (c2) at (5.5,1.75) {$$};
\node[fill=white] at (6,1.7) {\LARGE \bf $c_2^*$};
\node[point2] (f2) at (4.0,2.5) {$$};
\node[fill=white] at (4.7,2.8) {\LARGE \bf $f_2$};

\draw (f2o) edge[exedge] (c2);
\node[fill=white] at (6.25,0.9) {\LARGE \bf $\lambda_2^*$};
\draw (c2) edge[exedge, dotted] (f2);

\node[fcircle, minimum width=4.5cm] (c2) at (12.5,0) {$$};
\node[point1] (fko) at (12.5,0) {$$};
\node[fill=white] at (12.5,-0.5) {\LARGE \bf $f_k^*$};
\node[fill=white] at (12.5,-3.2) {\Huge \bf $C_k$};

\node[fill=white] at (12.5,4.3) {\huge \bf $\smallball_k$};
\node[dcircle, minimum width=3.8cm] (c1) at (12.2,1.8) {$$};
\node[point2] (ck) at (12.2,1.8) {$$};
\node[fill=white] at (11.8,1.5) {\LARGE \bf $c_k^*$};
\node[point2] (fk) at (11.1,3.4) {$$};
\node[fill=white] at (11.9,3.3) {\LARGE \bf $f_k$};

\draw (fko) edge[exedge] (ck);
\node[fill=white] at (12.8,1.0) {\LARGE \bf $\lambda_k^*$};
\draw (ck) edge[exedge, dotted] (fk);

\draw[loosely dotted, ultra thick, draw=black] (8.5,0) -- (10,0);
\end{tikzpicture}%
}
\caption{An illustration of facility selection for \fpt algorithm for solving
\kmedianpm instance.
}
\label{fig:mainfptapx}
\end{figure}

We will now focus on our main result, stated in
Theorem~\ref{theorem:mainfptapx}. As mentioned before, we build upon the
ideas for \kmedian of Cohen-Addad et al. of~\cite{cohen2019tight}. 
Their algorithm, however, does not apply
directly to our setting, as we have to ensure that the chosen
facilities satisfy the constraints.

A key observation is that by relying on the partition-matroid
constraint of the auxiliary submodular optimization problem, we can ensure that the
output solution will satisfy the \pattern. Since at least one \pattern
contains an optimal solution, we obtain the advertised approximation factor.

In the following lemma, we argue that this is indeed the case. Next, we
will provide an analysis of the running time of the algorithm. This
will complete the proof of Theorem~\ref{theorem:mainfptapx}.
\begin{lemma} 
\label{lemma:partition}
Let  $I=\divkins$ be an instance of \divkmedian to Algorithm~\ref{algo:divkmed}
and  $F^*=\{f_1^*, \dots, f_k^*\}$ be an optimal solution of $I$.
Let $J=((U,d), \{E_1^*,\cdots,E_k^*\},C',k)$ be an instance of \kmedianpm
corresponding to $F^*$, i.e,$f_i^* \in E^*_i, i \in [k]$.
On input $(J,\epsilon')$,
Algorithm~\ref{algo:kmedpm} outputs a set $\hat{S}$ satisfying $\cost(\hat{S})
\leq (1+\frac{2}{e}+\epsilon)\cost(F^*)$. Similarly, for \divkmeans, $\cost(\hat{S})
\leq (1+\frac{8}{e}+\epsilon)\cost(F^*)$.
\end{lemma}
\begin{proof}
Consider the iteration of Algorithm~\ref{algo:kmedpm} where the chosen clients
and radii are optimal, that is, $\lambda_i^*=d(c_i^*,f_i^*)$ and this distance
is minimal over all clients served by $f_i^*$ in the optimal solution.
Assuming the input described in the statement of the lemma, it is clear that in
this iteration we have $f_i^* \in \Pi_i$ (see Algorithm~\ref{algo:kmedpm},
line~\ref{algo:kmedpm:pi}).
Furthermore, given the partition-matroid constraint imposed on it, the proposed
submodular optimization scheme is guaranteed to pick exactly one facility from
each of $\Pi_i$, for all $i$.

On the other hand, known results for submodular optimization show that this
problem can be efficiently approximated to a factor within $\left(1-1/e\right)$
of the optimum~\cite{calinescu2011maximizing}. It is not difficult
to see this
translates into a $(1+\frac{2}{e}+\epsilon)$-approximation ($1+\frac{8}{e}+\epsilon$ resp.) of the optimal choice
of facilities, one from each of $\Pi_i$~\cite{cohen2019tight}. For complete calculations, please see Section \ref{app:fptapx}.
\end{proof}

\mpara{Running Time: }
First we bound the running time of Algorithm~\ref{algo:kmedpm}. Note that, the
runtime of  Algorithm~\ref{algo:kmedpm} is dominated by the two \textit{for}
loops (Line~2 and 3), since remaining steps, including finding approximate solution to the
submodular function \impr, runs in time $\poly(|U|)$. The \textit{for} loop of
clients (Line~2) takes time $ \bigO((k\nu^{-2} \log |U|)^k)$. Similarly, the \textit{for}
loop of discretized distances (Line~3) takes time $\bigO(([\Delta]_\eta)^k) =
\bigO(\log_{1+\eta}^k |U|)$, since $\Delta = \poly(|U|)$. Hence, setting $\eta = \Theta(\epsilon)$,
the overall
running time of Algorithm~\ref{algo:kmedpm} is bounded by\footnote{We use the
fact that, if $k \le \frac{\log |U|}{\log\log |U|}$, then $(k \log^2 |U|)^k =
\bigO(k^k \poly(n))$, otherwise if $k \ge \frac{\log |U|}{\log \log |U|}$, then
$(k \log^2 n)^k = \bigO(k^k (k \log k)^{2k})$.}
\[
\bigO\left(\left(\frac{k \log^{2} |U|}{\epsilon^2 \cdot \log(1+\epsilon)}\right)^k  \poly(|U|)\right) = \bigO\left(\left( \frac{ k^3 \log^2 k}{\epsilon^2
	\log(1+\epsilon)}\right)^k \poly(|U|) \right) 
\]
Since Algorithm~\ref{algo:divkmed} invokes Algorithm~\ref{algo:kmedpm} $\bigO (2^{tk})$ times, its running time is bounded by
$
\bigO\left(\left( \frac{2^t k^3 \log^2 k}{\epsilon^2
	\log(1+\epsilon)}\right)^k \poly(|U|) \right).
$ 



\begin{algorithm}
\caption{\sc Div-$k$-Med$(I=\divkins,\epsilon)$}
\footnotesize
\label{algo:divkmed}
\KwIn{$I$, an instance of the \divkmedian problem\\
\Indp \Indp ~$\epsilon$, a real number}

\KwOut{$T^*$, subset of facilities}

\ForEach{$\vec{\gamma} \in \{0,1\}^t$} {
    $E({\vec{\gamma}}) \gets \{f \in {F} : \vec{\gamma} = \vec{\chi}_{f}\}$
}
$\mathcal{E} \gets \{E({\vec{\gamma}}): \vec{\gamma} \in \{0,1\}^t \}$

$C' \leftarrow \textsc{coreset}((U,d),\pazocal{F},C,k, \nu \leftarrow \epsilon/16)$\\

$T^* \leftarrow \emptyset$\\
\ForEach{multiset $\{E(\vec{\gamma}_{1}),\cdots,E(\vec{\gamma}_{k})\} \subseteq \mathcal{E}$ of size $k$} {
  \If{$\sum_{i \in [k]}\vec{\gamma}_{i} \geq \vec{r}$, element-wise} {
    Duplicate facilities to make subsets in $\{E(\vec{\gamma}_{1}), \dots,
E(\vec{\gamma}_{k})\}$ disjoint\\
    $T \gets
\textsc{$k$-Median-PM}((U,d),\{E(\vec{\gamma}_{1}),\cdots,E(\vec{\gamma}_{k})\},C',\epsilon/4)$\\
    \If{$\textsf{cost}(C',T) < \textsf{cost}(C',T^*)$}{
      $T^* \gets T$\\
    }
  }
}
\Return{$T^*$}
\end{algorithm}

\section{Bicriteria approximation and heuristics}
\label{sec:bicri}
In this section, we describe a bicriteria
approximation algorithm that relies on a simple
observation: we can solve feasibility and clustering
independently, and merge the resulting solutions. 

First, we use a polynomial-time approximation algorithm $\mathcal{A}$ for \kmedian
ignoring the constraints in $\reqvec$ to obtain a solution. If the obtained
solution does not satisfy all the requirements in $\reqvec$. Then, we use a
feasibility algorithm $\mathcal{B}$, to obtain a feasible
{\pattern} of at most size $k$. Picking one facility in each $E_i$
of {\pattern} will satisfy our requirements in $\reqvec$.
Finally, we return the union of the two solutions, which have at most $2k$ facilities that
satisfy the lower-bound constraints and achieve the quality guarantee
of $\mathcal{A}$ (w.r.t. the optimal solution of size $k$).
We can employ, for instance, the local-search heuristic of Arya et
al., which yields a $(2k,3+\epsilon)$-approximation~\cite{arya2004local}, 
or the result of Byrka et al.~\cite{byrka2014improved} to achieve a
$(2k, 2.675)$-approximation. Recall from Proposition~\ref{proposition:introduction:3} 
that, even if we relax
the number of facilities to any function $f(k)$, it is unlikely to
approximate the \divkmedian problem in polynomial time. 
Our bicriteria approximation shows that this is not the case in \fpt time.

Leveraging the fact that $\mathcal{B}$ only needs to find one feasible
{\pattern}, instead of using the time-consuming Lemma~\ref{lemma:feasiblecp}, we
propose the following, relatively efficient strategy to obtain a feasibile
solution leading to exponential speedup.

\subsection{Dynamic programming approach (\DP)}
\begin{theorem}\label{thm:dpfeasibility}
There exists a deterministic algorithm with time $\bigO(kt2^t(r+1)^t \poly(|U|))$ that
can decide and find a feasible solution for \divkmedian. On the other hand,
assuming \seth, for every $\epsilon > 0$ there exists no algorithm that
decides the feasibility of \divkmedian in time 
$\bigO((r+1 - \epsilon)^t \poly(|U|))$ for every $r \ge 1$. 
Here $r = \max_{i \in [t]} r[i]$.
\end{theorem}
\begin{proof}
First we given an algorithm for feasibility.
An instance of feasibility problem is 
$I= (\mathcal{E}\subseteq \{0,1\}^t, \vec{f},\vec{r})$, where 
$f: \mathcal{E} \rightarrow [n]$ is the frequency vector, and 
$\vec{r} \in \{0,\cdots,k\}^t$ is the lower bound vector. 
We say a $k$-multiset $E$ of
$\mathcal{E}$ respects $f$, if for every $E_i \in E$, $E$ contains $E_i$ at
most $f(E_i)$ times. The goal is to find a $k$-multiset $E^*$ of $\mathcal{E}$
respecting $f$ such that $\sum_{i \in [k]} E^*_i \ge \vec{r}$.

The approach is similar to the dynamic program technique employed for {\sc
SetCover}. First, we obtain $\mathcal{E}'$ from $\mathcal{E}$ as
follows. For every element $E_i \in \mathcal{E}$, create $\min\{f(E_i),k\}$
copies of $E_i$ in $\mathcal{E}'$. Thus, $|\mathcal{E}'| \le k |\mathcal{E}|$.
Now let us arbitrarily order the elements in $\mathcal{E}'$ as 
$\mathcal{E}' = (E_1,E_2,\cdots)$. 
For every $i\in [|\mathcal{E}'|]$ and $\vec{\eta} \in \{0,\cdots,k\}^t$, we have an entry
$A[i,\eta]$ which is assigned the minimum number of elements in
$E_1,\cdots,E_i$ summing to at least $\vec{\eta}$. The dynamic program
recursion works as follows, as base case $A[0,\vec{0}] =0$. For each $ \vec{\eta} \ne \vec{0}$,
$A[0,\vec{\eta}] = \infty $
\begin{align}
  \label{eq:dp}
	A[i,\vec{\eta}] = & \min\{1+A[i-1,\vec{\eta} - E_i], A[i-1,\vec{\eta}]\}\\ 
	& \textit{ round negative entries in $\vec{\eta} - E_i$ to zero.} \nonumber
\end{align}

Finally, we check if $A[|\mathcal{E}'|,\vec{r}] \le k$. Note that any solution
on $\mathcal{E}'$ respects $f$ due to construction. Finally, the running time of
the above algorithm is $|\mathcal{E}'|\cdot (r+1)^t \cdot t =
\bigO(kt2^t(r+1)^t)$. 

To find a feasible solution, we update our dynamic program table as follows:	
For $A[0,\vec{0}] =''$. For each $ \vec{\eta} \ne \vec{0}$,
\begin{align*}
	A[0,\vec{\eta}] &= '00\ldots0' \quad \textit{string of $k+1$ zeros}  \\
	A[i,\vec{\eta}] &= 
	\begin{cases}
		A[i-1,\vec{\eta}] \quad \text{if } |A[i-1,\vec{\eta}]| < |\{E_i\} \cup A[i-1,\vec{\eta} - E_i]| \\
	\{E_i\} \cup  A[i-1,\vec{\eta} - E_i] \ \text{otherwise}
	\end{cases}
\end{align*}
$\textit{ round negative entries in $\vec{\eta} - E_i$ to zero, and union is for multiset.}$

Finally, we check whether or not $|A[|\mathcal{E}'|,\vec{r}]| \le k$.
%
To show the lower bound on runtime, note that if there exists an algorithm running
in time $\bigO((r+1 - \epsilon)^t \poly(|U|))$, for some $\epsilon >0$, then we
can solve {\sc SetCover} in time $\bigO((2-\epsilon)^{|\mathcal{U}|}
\poly(|\mathcal{U}|))$, where $\mathcal{U}$ is the universe of the
\textsc{SetCover} instance. This is because $r=1$ for \textsc{SetCover}, which
contradicts \seth~\cite[Conjecture~14.36]{cygan2015parameterized,cygan2016onproblems}.
\end{proof}

\subsection{Linear programming approach (\LP)}
In this subsection, we propose a heuristic for finding a feasible solution based
on randomized rounding of the fractional solution of a linear program. The linear
program formulation is as follows:
\begin{align*}
    \text{Minimize} \quad &  0 \cdot x~\text{such that}~\mathcal{E} \cdot x \geq \reqvec,
    \\  &  \sum_i x_i \leq k, 0 \leq x_i \leq f(E_i). 
\end{align*}
We solve the \LP to obtain a fractional solution and round $x_i$ to an
integer value using randomized rounding strategy inspired by~\cite{Raghavan1987rounding}.
$$
x'_i =
\begin{cases}
\lfloor{x_i}\rfloor & \text{with probability } 1 - x_i + \lfloor{x_i}\rfloor\\
\lceil{x_i}\rceil & \text{otherwise} \\
\end{cases}
$$
Therefore, it holds that $\mathbf{E}(\sum_i x'_i) \leq k$.
However, the lower requirements might be unfilled,
so the algorithm needs to verify the correctness and repeat the procedure a many times
and produce another solution (E.g., by randomizing objective function).

Even though the randomized rounding approach does not guarantee finding an
existing solution, the algorithm is very effective in real-world dataset, as
demonstrated in Section~\ref{sec:experiments}.

\subsection{Local search heuristic (\lsone)}
First, we present a local-search algorithm for \kmedianpm problem and discuss how
to apply this approach to solve \divkmedian problem. Given an
instance \kmedkpmins, we pick one facility from each $E_i$ at random as an
initial assignment, and continue to swap with facilities from the same group until
the solution is con verged i.e a facility $f \in E_i$ is only allowed to swap with
facility $f' \in E_i$ for all $i \in [k]$.

Recall that each feasible constraint pattern is an instance of \kmedianpm,
likewise, we employ the heuristic discussed above for each instance to obtain a
solution with minimum cost. The runtime of the algorithm is $\bigO(2^{tk}
\poly(|U|)$, since we have at most $\bigO(2^{tk})$ feasible constraint patterns
and each iteration of the local search can be executed in polynomial time. In our
experiments, we refer to this algorithm as \lsone. Bounding the approximation ratio
of \lsone is left as an open problem.

\section{Experiments}
\label{sec:experiments}
This section discusses our experimental setup and results. Our objective is
mainly to evaluate the scalability of the proposed methods.

\subsection{Experimental setup}

\mpara{Hardware.}
Our experiments make use of two hardware configurations: ($i$)
a {\em desktop} with a $4$-core {\em Haswell} CPU with $16$~GB of main memory,
Ubuntu 21.10; ($ii$) a {\em compute node} with a $20$-core {\em Cascade lake}
CPU with $64$~GB of main memory, Ubuntu 20.04.

\mpara{Datasets.}
We use datasets from UCI machine learning repository~\cite{dua2019uci} (for details check
the corresponding webpage of each dataset). Data are processed by assigning
integer values to categorical data and normalize each column to unit norm. We
assume the set of clients and facilities to be the same i.e,, $U = C = F$.

\mpara{Data generator.} 
For scalability, we generate synthetic data using {\tt make\_blob}
from {\tt scikit-learn}. The groups are generated by sampling data points
uniformly at random and restricting the maximum groups a data point can belong
to, i.e. each data point belongs to at least one and at most $t/2$ groups. 
We ensure that the same instance is generated for each
configuration by using an initialization seed value.

\mpara{Baseline.} For \divkmedian, we use the local-search algorithm with no
requirement constraints as a baseline, denoted as \lszero, which is a
$5$-approximation for \kmedian. Additionally, we implemented a
$p$-swap local-search algorithm, denoted as \lszeropswap, which is a
$(3+\frac{2}{p})$-approximation~\cite{arya2001local}. 
We observed no significant improvement in the cost of solution of
\lszeropswap compared to \lszero with $p=2$.
We also experimented with trivial algorithms based
on {\em brute force} and {\em linear program} solvers, which fail to scale
for even modest size instance with $|U|=100, k= 6$.
For \divkmeans, we use {\tt $k$-means++} with no requirement constraints as
baseline, denoted as \KM, which is a $\log(k)$-approximation
for \kmeans~\cite{vassilvitskii2006kmeans}. 
For each dataset we
perform $5$ executions of \lszero (or \KM) using random initial assignments to obtain
a solution with minimum cost. 

\mpara{Implementation.} 
Our implementation is written in {\tt Python}
programming language. We make use of {\tt numpy} and {\tt scipy} python packages for
matrix computations. We use {\tt $k$-means++} implementation from 
{\tt scikit-learn}. For coresets we use importance sampling, which results in
coresets of size $\bigO(kD\epsilon^{-2})$,
$\bigO(kD^3\epsilon^{-2})$ for \kmedian and \kmeans, respectively, where $D$ is
the dimension of data~\cite{feldman2011unified, bachem2017practical}.
For discretizing distances we use the existing implementation of binning from 
{\tt scikit-gstat} python package.

Our exhaustive search algorithm is implemented as matrix multiplication
operation, thereby we use optimized implementation of {\tt numpy} to
enumerate feasible constrained patterns. To achieve this, we generate two
matrices, a matrix $A_{|\charpart| \times t}$
encoding bit vectors $[0,1]^t$ corresponding to subset lattice of
$\groups=\{G_1,\dots,G_t\}$, and,
matrix $B_{{{|\charpart|}\choose{k}} \times |\charpart|}$ enumerating all combinations (with
repetitions) of choosing $k$ facilities from non-intersecting groups in
$\charpart$ and multiply $B \times A = R_{{{|\charpart|}\choose{k}} \times t}$.
Finally, for each row of $R$
we verify if the requirements in $\reqvec$ are satisfied elementwise to obtain
of all feasible constraint patterns. More precisely, if $i$-th row of $R$
satisfy requirements in $\reqvec$ elementwise, it implies that $i$-th row of $B$ is a feasible
constraint pattern. For finding one feasible
constrained pattern, we enumerate rows of the matrix $B$ in batches, $p$ rows
at a time and multiply with matrix $A$ until we obtain one solution satisfying
$\reqvec$ elementwise. This is
essential for scaling of bicriteria algorithms, where it is sufficient to find one feasible
constraint pattern.
Our dynamic program is implemented using {\tt numpy} arrays. Observe
carefully in Equation~(\ref{eq:dp}) that computing $A[i, \vec{\eta}]$ relies
only on the values of $A[i-1,\vec{\eta}]$, so we only use
$2\cdot(r+1)^t k$ memory.
To reduce memory footprint and improve scalability, we avoid
precomputing of distances between datapoints, which requires $\bigO(n^2)$ memory.
Instead, we compute distances on-the-fly. Nevertheless, this has an additional overhead
of $\bigO(D)$ in time.

Our implementation is available anonymously as open source~\cite{sourcecode}.

\subsection{Experimental results}
This subsection discusses our experimental results.

\begin{table}
\caption{\label{table:experiments:setup}Experimental setup. Data dimension $D=5$.}
\footnotesize
\begin{tabular}{l r r r}
\toprule
Experiment & $|U|$ & $t$ & $k$ \\
\midrule
\multicolumn{3}{l}{\underline{Figure~\ref{fig:scaling:feasibility}~(Feasibility)}}\\
  left & $10^3, \dots, 10^8$ &          $7$ &          $5$ \\
center &              $10^4$ &          $6$ & $4,\dots,20$ \\
 right &              $10^4$ & $4,\dots,14$ &          $5$ \\
\midrule
\multicolumn{3}{l}{\underline{Figure~\ref{fig:scaling:bicriteria:1},\ref{fig:scaling:bicriteria:2}~(Bicriteria)}}\\
  left & $10^2,\dots,10^5$ & $7$ & $5$ \\
center & $10^4$ & $5$ & $4,\dots,10$ \\
 right & $10^4$ & $4,\dots,8$ & $4,\dots,10$ \\
\midrule
\multicolumn{3}{l}{\underline{Figure~\ref{fig:scaling:es-ls}~($\ES + \lsone$)}}\\
  left & $10^2,4\cdot10^2,\dots,10^4,4\cdot10^4$ & $5$ & $4$\\
center & $10^3$ & $5$ & $2,\dots,7$ \\
 right & $10^3$ & $4,\dots,7$ & $5$ \\
\bottomrule
\end{tabular}
\vspace{-0.5cm}
\end{table}

\mpara{Scalability.}
The experimental setup of our scalability experiments is available in
Table~\ref{table:experiments:setup}. Our feasibility experiments execute on
the {\em desktop} configuration. Bicriteria experiments are executed on the {\em compute}
node. All reported runtimes are in seconds, and we terminated experiments that
took more than two hours. 

Our first set of experiments studies the scalability of finding one feasible
constraint pattern. In Figure~\ref{fig:scaling:feasibility}, we report the runtime of exhaustive
enumeration (\ES), dynamic program (\DP) and linear program (\LP) algorithms for
finding a feasible constraint satisfaction pattern as a function of number of
facilities $|U|$ (left), number of cluster centers $k$ (center) and number of
groups $t$ (right). 
For each configuration of $|U|$, $t$ and $k$ in
Table~\ref{table:experiments:setup}, we report runtimes of 10 independent input
instances. We observed little variance in runtime among the independent inputs
for \DP and no significant variance in runtime for \LP and \ES.
Recall that finding a feasible constraint pattern is \np-hard and \wtwo-hard
(See Section~\ref{sec:hardness}). Despite that, our algorithms solve instances with up to 
$100$ million facilities
in less than one hour on a desktop computer, provided that number of
cluster centers and groups are small i.e, $k=5, t=7$.

Surprisingly, \LP performs better with respect to runtime. However,
randomized rounding fails find a feasible solution for large
values of $t > 8$, since the fractional solution obtained is sparse when
$|\charpart| \approx 2^t > 2^8$. Additionally, when
the number of facilities is large i.e, $|U|=10^8$, 
the runtime of \LP, \DP and \ES converge.

\begin{figure*}
\arraycolsep=0.0pt\def\arraystretch{0.0}
\captionlistentry{}
\label{fig:scaling:feasibility}
\captionlistentry{}
\label{fig:scaling:bicriteria:1}
\captionlistentry{}
\label{fig:scaling:es-ls}

\begin{tabular}{c}
\includegraphics[width=0.8\linewidth]{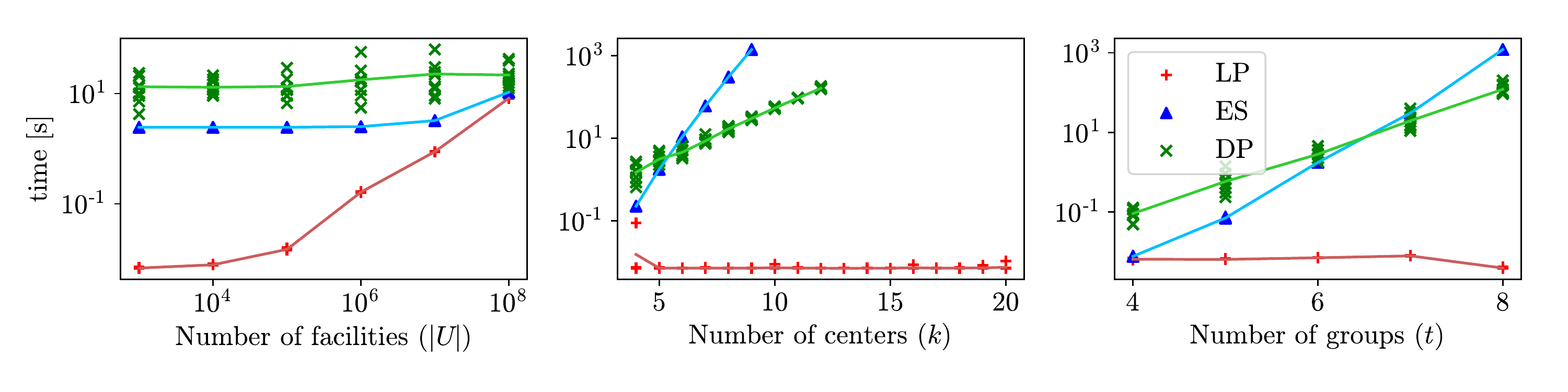}\\
{\bf Figure~2: Scalability of algorithms for finding feasible constraint pattern.}\\
\includegraphics[width=0.8\linewidth]{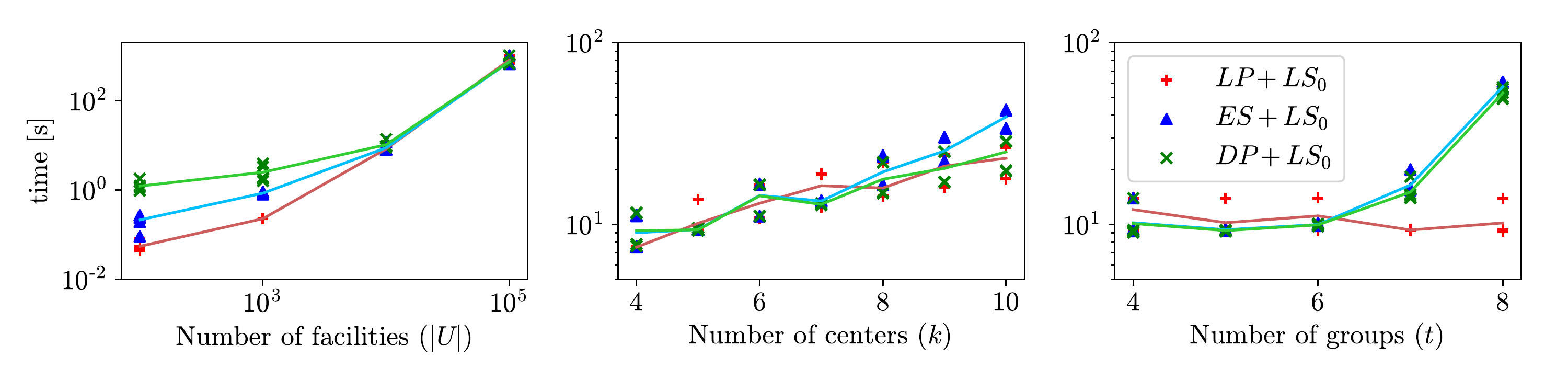}\\
{\bf Figure~3: Scalability of bicriteria algorithms for \divkmedian.}\\
\includegraphics[width=0.78\linewidth]{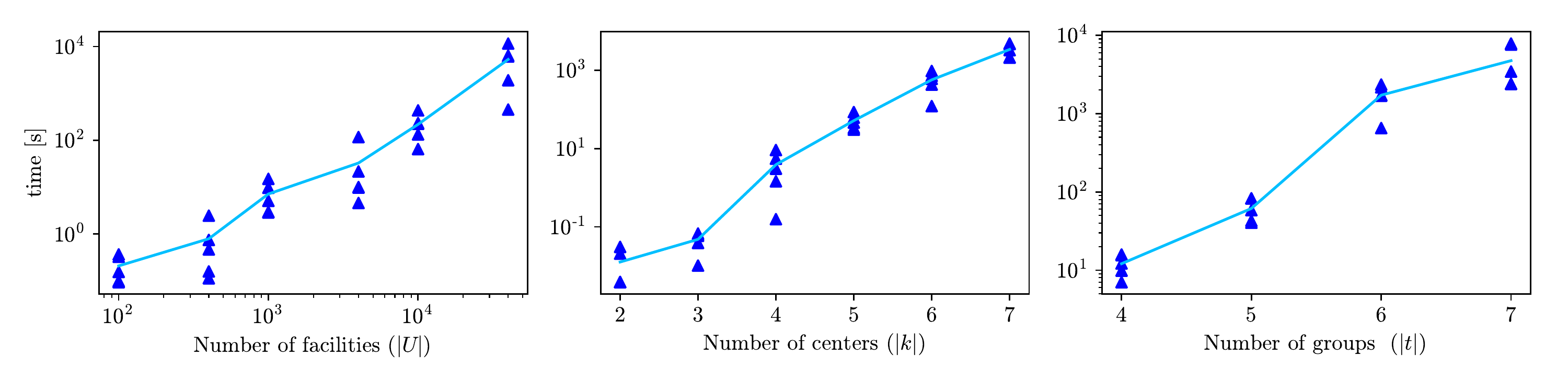}\\
{\bf Figure~4: Scalability of $\ES+\lsone$ algorithm for \divkmedian.}
\end{tabular}
\end{figure*}


Our second set of experiments studies the scalability of bicriteria algorithms.
In Figure~\ref{fig:scaling:bicriteria:1}, we report the runtime of local search
(\lszero) combined with \ES, \DP and \LP algorithms for solving \divkmedian
problem, as a function of number of facilities $|U|$ (left), number of cluster
centers $k$ (center) and number of groups $t$ (right). 
For each configuration of $|U|$, $t$ and $k$ in
Table~\ref{table:experiments:setup}, we report runtimes of $5$ independent input
instances and observed little variance in runtime.
We observed similar scalability for \divkmeans, for which we
make use of {\tt kmeans++} (\KM) implementation from {\tt scikit-learn}, along with
\ES, \LP and \DP (See Supplmentary~\ref{fig:scaling:bicriteria:2}).

Our third set of experiments studies the scalability of $\ES + \lsone$
In Figure~\ref{fig:scaling:es-ls}, we report runtime as a function of number of
facilities $|U|$ (left), number of centers $k$ (center) and number of groups $t$
(right). 
For each configuration of $|U|$, $t$ and $k$ in
Table~\ref{table:experiments:setup}, we report runtimes of $5$ independent input
instances and observed little variance in runtime.
We observed high variance in runtime, as a result of variance in the number of
feasible constraint patterns among independent inputs. The algorithm manages to
solve two instances with up to 40 thousand facilities in approximately $2.5$ hours
on a desktop computer.



\mpara{Experiments on real datasets.}
For each dataset, we generate two disjoint groups $G_1, G_2$. For this we choose
gender, except for {\tt house-votes}, where we choose party affiliation. We use the
protected attributes race or age group to generate groups $G_3$ and $G_4$,
respectively, so groups $G_3, G_4$ intersect with either or both $G_1, G_2$.
The experiments are executed on {\em desktop} with number of centers $k=6$, number of groups
$t=4$ and requirement vector $\reqvec=\{3,3,2,1\}$. That is, we have a
requirement that the chosen cluster centers must be an equal number of men and
women, with additional requirements to pick at least
two cluster centers that belong to a group representing race and one
center that belongs to a group representing a certain age group. 
For each dataset, we execute $5$ iterations of each algorithm with different
initial assignment to report a solution with minimum cost and corresponding runtime.

In Table~\ref{table:experiments:dataset}, we report dataset name, size $|U|$, 
dimension $D$ in Column~1-3, respectively. Column~4 is the
runtime of \lszero (baseline). We report results of bicriteria
approximation algorithms $\lszero+\LP$ in Column~5-7, $\lszero+\ES$ in Column~8-10 and
$\lszero+\DP$ in Column~11-13. For each bicriteria algorithm we report runtime,
$\zeta^*=\frac{\cost(ALG)}{\cost(\lszero)}$ which is the ratio of the cost bicriteria
algorithm to the cost of \lszero and the size of reported solution $k^*$. Finally, in Column~14-15,
Column~16-17 and Column~18-19, we report results of $\LP+\lsone$, $\ES+\lsone$ and
$\fpt(3 + \epsilon)$ approximation algorithm, respectively. For each of these algorithms we
report runtime and $\zeta^*=\frac{\cost(ALG)}{\cost(\lszero)}$.


In bicriteria algorithms, \lszero consumes the majority ($> 90\%$) of the runtime, and a minority ($<10\%$) of the runtime is spent on finding a feasible
constraint pattern. This observation is trivial by comparing runtime of
bicriteria algorithm(s) and \lszero. 
As expected, $\lszero+\DP$ returns solution with minimum size $k^*$ with no
significant change in the cost of solution obtained from \LP+\lszero and \ES+\lszero.
Note that the value of $\zeta^* < 1.0$ since the size of solution
obtained $k^* > k = 6$.




Even though the \fpt approximation algorithms presented in Section~\ref{sec:algorithm} have good
theoretical guarantees, they fail to perform in practice. We believe the reason
is two-fold. First, the size of the coreset obtained using importance
sampling $\left(\bigO(kD^2\epsilon^{-2})\right)$ is relatively large. Second, the
$\epsilon'$ factor used for discretizing distances is also large. In this regard,
there is still room for implementation engineering to make the algorithm
practically viable.

\begin{table*}
\caption{\label{table:experiments:dataset}Experiments on datasets with $k=6$,
$t=4$ and $\reqvec=\{3,3,2,1\}$.}
\footnotesize
\begin{tabular}{l r r r r r r r r r r r r r r r r r r r}
\toprule
& & & & 
\multicolumn{9}{c}{Bicriteria approximation ($2k, \alpha$)} &
\multicolumn{4}{c}{Heuristics ($k$)} &
\multicolumn{2}{c}{\fpt($k,t,\epsilon$)}\\

\cmidrule(r{1em}){5-13}
\cmidrule(r{1em}){14-17}
\cmidrule(r{1em}){18-19}

& & & \lszero &
\multicolumn{3}{c}{\lszero + \LP} & 
\multicolumn{3}{c}{\lszero + \ES} & 
\multicolumn{3}{c}{\lszero + \DP} &
\multicolumn{2}{c}{\LP + \lsone} & 
\multicolumn{2}{c}{\ES + \lsone} & 
\multicolumn{2}{c}{$(3+\epsilon)$-apx} \\

\cmidrule(r{0.75em}){5-7}
\cmidrule(r{0.75em}){8-10}
\cmidrule(r{0.75em}){11-13}
\cmidrule(r{0.75em}){14-15}
\cmidrule(r{0.75em}){16-17}
\cmidrule(r{0.75em}){18-19}

Dataset & $|U|$ & $D$ & 
time & 
time & $\zeta^*$ & $k^*$ &
time & $\zeta^*$ & $k^*$ & 
time & $\zeta^*$ & $k^*$ &
time & $\zeta^*$ & 
time & $\zeta^*$ & 
time & $\zeta^*$ \\
\midrule
				 switzerland &   123 & 14 &   0.05 &   0.14 & 0.92 & 10 &   0.05 & 0.92 & 10 &   0.09 & 0.92 & 10 &   0.35 & 1.08 &     0.16 & 1.08 & 16\,841.32 & 2.82 \\
					 hepatitis &   155 & 20 &   0.07 &   0.07 & 0.94 & 11 &   0.07 & 0.95 & 10 &   0.11 & 0.95 & 10 &   0.39 & 1.07 &     0.27 & 1.07 & 18\,922.51 & 1.81 \\
									va &   200 & 14 &   0.06 &   0.06 & 0.95 & 11 &   0.06 & 0.95 & 11 &   0.10 & 0.98 &  9 &   0.20 & 1.27 &     0.01 & 1.27 & 14\,855.96 & 1.76 \\
					 hungarian &   294 & 14 &   0.14 &   0.14 & 0.95 & 10 &   0.14 & 0.96 &  9 &   0.17 & 0.98 &  8 &   0.74 & 1.02 &     4.00 & 1.01 &          - &  - \\
			 heart-failure &   299 & 13 &   0.18 &   0.19 & 0.93 & 11 &   0.19 & 0.95 &  9 &   0.22 & 0.95 &  9 &   0.71 & 1.05 &     3.72 & 1.05 &          - &  - \\
					 cleveland &   303 & 14 &   0.09 &   0.10 & 0.93 & 10 &   0.10 & 0.99 &  9 &   0.13 & 0.99 &  8 &   0.47 & 1.07 &     1.33 & 1.05 &          - &  - \\
				 student-mat &   395 & 33 &   0.24 &   0.25 & 0.96 & 12 &   0.25 & 0.97 & 12 &   0.28 & 0.99 &  8 &   0.36 & 1.05 &     0.32 & 1.05 &          - &  - \\
			house-votes-84 &   435 & 17 &   0.16 &   0.16 & 0.97 & 10 &   0.16 & 0.98 &  9 &   0.19 & 0.98 &  9 &   0.71 & 1.17 &     3.20 & 1.11 &          - &  - \\
				 student-por &   649 & 33 &   0.50 &   0.51 & 0.98 & 10 &   0.50 & 0.98 & 10 &   0.53 & 0.99 &  9 &   0.49 & 1.02 &     0.52 & 1.02 &          - &  - \\
		drug-consumption &  1884 & 32 &   2.58 &   2.69 & 0.98 & 12 &   2.68 & 0.98 & 12 &   2.72 & 0.99 &  8 &   0.49 & 1.08 &     0.41 & 1.07 &          - &  - \\
								bank &  4521 & 17 &   8.56 &   8.72 & 0.97 & 10 &   8.71 & 0.99 & 10 &   8.76 & 0.98 &  9 &   1.41 & 1.10 &     2.07 & 1.10 &          - &  - \\
						 nursery & 12960 &  9 &  40.21 &  40.48 & 0.99 & 10 &  40.66 & 0.99 & 10 &  40.43 & 0.99 &  9 &  22.38 & 1.14 &    43.20 & 1.14 &          - &  - \\
			vehicle-coupon & 12684 & 26 &  51.87 &  51.34 & 0.98 & 12 &  50.88 & 0.98 & 12 &  50.98 & 0.99 &  8 &   8.59 & 1.12 &    16.43 & 1.12 &          - &  - \\
				 credit-card & 30000 & 25 & 928.77 & 945.56 & 0.99 & 12 & 939.98 & 0.99 & 12 & 941.07 & 1.00 &  8 &   9.18 & 1.18 &    18.89 & 1.18 &          - &  - \\
				dutch-census & 32561 & 15 & 376.73 & 384.15 & 0.97 & 12 & 390.82 & 0.98 & 12 & 385.36 & 0.99 &  8 &  76.34 & 1.40 &   151.18 & 1.32 &          - &  - \\
					 bank-full & 45211 & 17 & 934.14 & 958.79 & 0.97 & 11 & 958.86 & 0.98 & 11 & 948.85 & 0.97 & 10 & 103.57 & 1.10 &   202.73 & 1.10 &          - &  - \\
     diabetes & 101\,766 & 50 & 15\,896.14 &      - &    - &  - &      - &    - &  - &      - &    - &  - & 829.96 & 1.07 &1\,503.05 & 1.01 &          - &  - \\
\bottomrule
\end{tabular}
\vspace{-0.4cm}
\end{table*}

\section{Related work}
\label{sec:related}
\kmedian is a classic problem in computer science.
The first constant-factor approximation for metric \kmedian was presented by
Charikar et al.~\cite{charikar2002constant}, which was improved to
$(3+\epsilon)$ in a now seminal work by Arya et al.~\cite{arya2004local}, using
a local-search heuristic. The best known approximation ratio for metric
instances stands at 2.675, which is due to
Byrka et al.~\cite{byrka2014improved}. Kanungo et al.~\cite{kanungo2004local}
gave a $(9+ \epsilon)$ approximation algorithm for \kmeans, 
which was recently improved to $6.357$ by Ahmadian et al.\cite{ahmadian2019better}.
On the other side of the coin, the \kmedian
problem is known to be \np-hard to approximate to a factor less than
$1+2/e$~\cite{guha1998greedy}. Bridging this gap is a well known open problem.
In the \fpt landscape, finding an optimal
solution for \kmedian/\kmeans are known to be \wtwo-hard with respect to parameter $k$
due to a reduction by Guha and Khuller~\cite{guha1998greedy}. More recently,
Cohen-Addad et al.~\cite{cohen2019tight} presented \fpt approximation algorithms
with respect to parameter $k$, with approximation ratio
$(1+\frac{2}{e}+\epsilon)$ and $(1+\frac{8}{e}+ \epsilon)$ for \kmedian and
\kmeans, respectively. They showed that the ratio is essentially tight assuming \gapeth. Their
result also implies a $(2+\epsilon)$ approximation algorithm for \matmedian in
\fpt time with respect to parameter $k$.

In recent years the attention has turned in part to variants of the problem with
constraints on the solution. One such variant is the red-blue median problem
(\rbmedian), in which the facilities are colored either red or blue, and a
solution may contain only up to a specified number of facilities of each
color~\cite{hajiaghayi2010budgeted}. This formulation was generalized by the
matroid-median problem (\matmedian)~\cite{krishnaswamy2011matroid}, where
solutions must be independent sets of a matroid. 
Constant-factor approximation algorithms were given by 
Hajiaghayi et al.~\cite{hajiaghayi2010budgeted,hajiaghayi2012local} and
Krishnaswamy et al.~\cite{krishnaswamy2011matroid} for \rbmedian and \matmedian
problems, respectively.

\mpara{Algorithmic fairness.} In recent years, the notion of fairness
in algorithm design has gained significant traction. The underlying
premise concerns data sets in which different social groups, such as
ethnicities or people from different socioeconomic backgrounds,  can be identified.
The output of an algorithm, while suitable when measured
by a given objective function, might negatively impact
one of said groups in a disproportionate
manner~\cite{biddle2020predicting,berk2021fairness}.

In order to mitigate this shortcoming, constraints or penalties can be
imposed on the objective to be optimized, so as to promote more
equitable
outcomes~\cite{hardt2016equality,zafar2017fairness,dwork2018decoupled,fu2020fairness}.

In the context of clustering, which is the focus of the present work,
existing proposals have generally dealt with the notion of equal
representation within
clusters~\cite{chierichetti2017fair,rosner2018privacy,schmidt2019fair,huang2019coresets,backurs2019scalable,bercea2019cost}. That
is, the \textit{clients} in each cluster should not be comprised
disproportionately of any particular group, and all groups should
enjoy sufficient representation in all clusters. In contrast, this paper deals with
the problem of representation constraints among the
\textit{facilities}, as formalized in the recent work of Thejaswi et
al~\cite{thejaswi2021diversity}. 
While the
problem admits no polynomial-time approximation algorithms for the
general case, the authors of the original work presented constant-factor
approximation algorithms for special cases~\cite{thejaswi2021diversity}.

\section{Conclusions \& future work}
\label{sec:conclusions}
In this paper we have provided a comprehensive analysis of the \textit{diversity-aware $k$-median}
problem, a recently proposed variant of \kmedian in the context of
algorithmic fairness. We have provided a thorough characterization of
the parameterized complexity of the problem, as well as the first
fixed-parameter tractable constant-factor approximation algorithm. Our
algorithmic and hardness results naturally extend for {diversity-aware
$k$-means} problem.

Despite its theoretical guarantees, said approach is
impractical. Thus, we have proposed a faster dynamic program for
solving the feasibility problem, as well as an efficient, practical
linear-programming approach to serve as a building block for the
design of effective heuristics. In a variety of experiments on
real-world and synthetic data, we have shown that our approaches are
effective in a wide range of practical scenarios, and scale to
reasonably large data sets.

Our results open up several interesting directions for future
work. For instance, it remains unclear whether further speed-ups are
possible in the solution of the feasibility problem. The exponent
of $t\log r$ in our algorithm is close to the known lower bound of
$t$, so it is natural to ask whether this extra factor can be shaved off.

\bibliographystyle{ACM-Reference-Format}
\bibliography{sources-short}

\appendix
\clearpage
\section{Proofs} \label{app:proofs}

\subsection{Scalability experiments}
In Figure~\ref{fig:scaling:bicriteria:2} we report the scalability of bicriteria approximation
algorithms for \divkmeans problem.

\begin{figure*}
\includegraphics[width=0.8\linewidth]{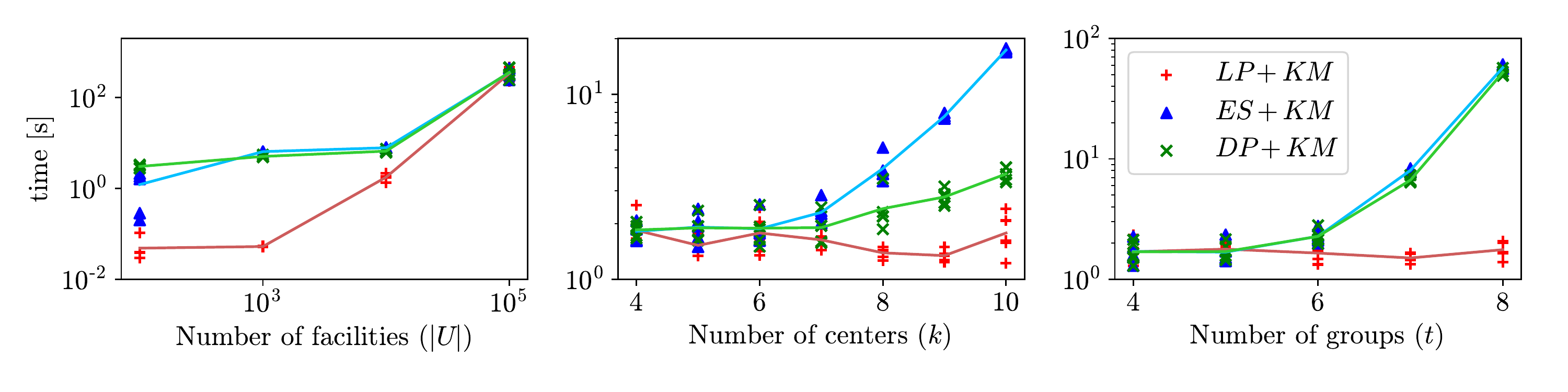}
\caption{\label{fig:scaling:bicriteria:2}Scalability of bicriteria algorithm for \divkmeans.}
\end{figure*}

\subsection{Proof of Lemma~\ref{lemma:partition}} \label{app:fptapx}

In fact, we prove Theorem~\ref{theorem:mainfptapx}. We primarily focus on
\divkmedian, indicating the parts of the proof for \divkmeans. In essence, to
achieve the results for \divkmeans, we need to consider squared distances which
results in the claimed approximation ratio with same runtime bounds.

As mentioned before, to get a better approximation factor, the idea is to reduce
the problem of finding an optimal solution to \kmedianpm to the problem of
maximizing a monotone submodular function. To this end, for each $S \subseteq
\pazocal{F}$, we define the submodular function $\impr(S)$ that, in a way,
captures the cost of selecting $S$ as our solution. To define the function
$\impr$, we add a fictitious facility $F'_i$, for each $i \in [k]$ such that
$F'_i$ is at a distance $2\lambda^*_i$ for each facility in $\Pi_i$. We, then,
use the triangle inequality to compute the distance of $F'_i$ to all other
nodes.
Then, using an $(1-1/e)$-approximation algorithm (Line~12), we approximate
$\impr$. Finally, we return the set that has the minimum \kmedian cost over all
iterations.

\begin{algorithm}
\footnotesize
\caption{\sc ${\kmedianpm(J=((U,d),\{E_1,\cdots,E_k\},C'), \epsilon')}$}
\label{algo:kmedpm}
\KwIn{$J$, an instance of the \kmedianpm problem\\
\Indp \Indp ~$\epsilon'$, a real number}

\KwOut{$S^*$, a subset of facilities}

$S^* \leftarrow \emptyset$,
$\eta \leftarrow e\epsilon'/2$\\
\ForEach{ordered multiset $\{c'_{i_1},\cdots, c'_{i_k}\} \subseteq C'$ of size $k$} {
  \ForEach{ordered multiset $\Lambda = \{\lambda_{i_1},\cdots, \lambda_{i_k}\}$ such that $ \lambda_{i_1} \subseteq [[\Delta]_\eta]$} {
    \For{$j=1$ to $k$} {
      $\Pi_j \leftarrow \{f \in E'_j \mid d_D(f,c'_{i_j}) = \lambda_{i_j} \}$ \label{algo:kmedpm:pi}\\
      Add a fictitious facility $F'_j$ \\
      \For{$f \in \Pi_j$} {
        $d(F'_j,f) \gets 2\lambda_{i_j}$
      }
      \For{$f \notin \Pi_j$} {
        $d(F'_j,v) \gets 2\lambda_{i_j} + \min_{f \in \Pi_j} d(f,v)$
      }
    }
    \For{$S \subseteq F$, define $\textsf{improve}(S) := \textsf{cost}(C',F') - \textsf{cost}(C',F' \cup S)$} {
      $S_{max} \leftarrow S \subseteq  \pazocal{F}$ that maximizes
      $\textsf{improve}(S)$ s.t. $|S \cap \Pi_j|=1, \forall j\in[k]$\\
      \If{$\textsf{cost}(C',S_{max}) < \textsf{cost}(C',S^*)$} {
        $S^* \leftarrow S$
      }
    }
  }
}
\Return $S^*$
\end{algorithm}

\textbf{Correctness: }
Given $I=\divkins$ . Let $\mathcal{J} :=\{J_\ell\}$ be the instances of
\kmedianpm generated by Algorithm~\ref{algo:divkmed} at Line~9 (For simplicity,
we do not consider client coreset  $C'$ here). For correctness, we show that
$F^* = \{f^*_i\}_{i \in [k]} \subseteq \pazocal{F}$ is feasible to $I$ if and
only if there exists $J_{\ell^*} =  ((V,d),\{E^*_{1},\cdots, E^*_{k}\},C) \in
\mathcal{J}$ such that $F^*$ is feasible to $J_{\ell^*}$. This is sufficient,
since the objective function of both the problems is same, and hence returning
minimum over $\mathcal{J}$ of an optimal (approximate) solution obtains an
optimal (approximate resp.) solution to $I$.
We need the following proposition for the proof.
\begin{proposition} \label{cl:setfacvec}
For all $E \in \mathcal{E}$, and for all $f \in E$, we have $\vec{\chi}_f = \vec{\chi}_E$.
\end{proposition}
\begin{proof}
Fix $E \in \mathcal{E}$ and $f \in E$. Since $E \in \mathcal{E}$, there exists
$\vec{\gamma} \in \{0,1\}^t$ such that $E_{\vec{\gamma}} = E$. But this means
$\vec{\chi}_{E} = \vec{\gamma}$. On the other hand, since $f \in
E_{\vec{\gamma}}$, we have that $\vec{\chi}_f = \vec{\gamma}$.
\end{proof}
Suppose $F^* = \{f^*_i\}_{i \in [k]} \subseteq \pazocal{F}$ is a feasible solution to $I: $
$
\sum_{i \in [k]} \vec{\chi}_{f^*_i} \ge \vec{r}.
$
Then, consider the instance $J_{\ell^*} =  ((V,d),\{E^*_{1},\cdots,
E^*_{k}\},C)$  with $E^*_i = E_{\vec{\chi}_{f^*_i}}$, for all $i \in [k]$.
Since, 
\[
\sum_{i \in [k]} \vec{\chi}_{E^*_i} = \sum_{i \in [k]} \vec{\chi}_{f^*_i} \ge \vec{r}
\]
we have that $J_{\ell^*} \in  \mathcal{J}$.
Further, $F^*$ is feasible to $J_{\ell^*}$ since $F^* \cap E^*_{i} = f^*_i$ for all $i \in [k]$.
For the other direction, fix an instance $J_{\ell^*} =  ((V,d),\{E^*_{1},\cdots,
E^*_{k}\},C) \in \mathcal{J}$ and a feasible solution $F^* =\{f^*_i\}_{i \in
[k]}$ for $J_{\ell^*}$. From Claim~\ref{cl:setfacvec} and the feasiblity of
$F^*$, we have $\vec{\chi}_{f^*_i} = \vec{\chi}_{E^*_i}$. Hence,
\[
 \sum_{i \in [k]} \vec{\chi}_{f^*_i} = \sum_{i \in [k]} \vec{\chi}_{E^*_i}  \ge \vec{r}
\]
which implies $F^* =\{f^*_i\}_{i \in [k]}$ is a feasible solution to $I$. To
complete the proof, we need to show that the distance function defined in
Line~10 is a metric, and \impr function defined in Line~11 is a monotone
submodular function. Both these proofs are the same as that in
\cite{cohen2019tight}. 

\textbf{Approximation Factor: } For $I=\divkins$, let
$I'=((U,d),C',F,\mathcal{G}, \vec{r},k)$  be the instance with client coreset
$C'$. 
Let $F^* \subseteq \pazocal{F}$ be an optimal solution to $I$, and let
$\Tilde{F}^* \subseteq \pazocal{F}$ be an optimal solution to $I'$. Then, from
\coreset Lemma~\ref{sec:algoritheorem:coresets}, we have that
\[
(1-\nu) \cdot \textsf{cost}(F^*,C) \le \textsf{cost}(F^*,C') \le (1+\nu) \cdot \textsf{cost}(F^*,C).
\]
The following proposition, whose proof closely follows that
in~\cite{cohen2019tight}, bounds the approximation factor of
Algorithm~\ref{algo:divkmed}.

\begin{proposition} \label{lem:fptapxpm}
For $\epsilon' >0$, let $J_\ell=((V,d),\{E_1,\cdots,E_k\},C',\epsilon')$ be an input to Algorithm~\ref{algo:kmedpm}, and let $S^*_\ell$ be the set returned. Then,
\[
\textsf{cost}(C',S^*_\ell) \le (1+2/e + \epsilon') \cdot \textsf{OPT}(J_\ell),
\]
where $\textsf{OPT}(J_\ell)$ is the optimal cost of \kmedianpm on $J_\ell$. Similarly, for \divkmeans,
\[
\textsf{cost}(C',S^*_\ell) \le (1+8/e + \epsilon') \cdot \textsf{OPT}(J_\ell),
\]
\end{proposition}
This allows us to bound the approximation factor of Algorithm~\ref{algo:divkmed}.
\begin{align*}
\textsf{cost}(T^*, C') 
=   \min_{J_\ell \in \mathcal{J}}  \textsf{cost}(S^*_\ell,C') 
\le  (1+2/e+\epsilon') \cdot \textsf{cost}(\Tilde{F}^*,C') \\
\le (1+2/e +\epsilon') \cdot \textsf{cost}(F^*,C')
\le (1+2/e +\epsilon')(1+ \nu) \cdot \textsf{cost}(F^*,C).
\end{align*}
On the other hand, we have $\textsf{cost}(T^*,C) \le (1+2\nu) \cdot \textsf{cost}(T^*, C')$. Hence, using $\epsilon'=\epsilon/4$ and $\nu = \epsilon/16$, we have
\begin{align*}
\textsf{cost}(T^*,C) &\le (1+2/e +\epsilon')(1+ \nu) (1+2\nu) \cdot \textsf{cost}(F^*,C) \\
&\le (1+2/e+\epsilon)\cdot \textsf{cost}(F^*,C)
\end{align*}
for $\epsilon\le 1/2$. Analogous calculations holds for \divkmeans. This finishes the proof of Lemma~\ref{lemma:partition}.
Now, we prove Proposition~\ref{lem:fptapxpm}.

\begin{proof}
Let $F^*_\ell$ be an optimal solution to $J_\ell$. Then, since
$\textsf{cost}(C',F'\cup F^*_\ell) = \textsf{cost}(C',F^*_\ell)$, we have that
$F^*_\ell$ is a maximizer of the function $\impr(\cdot)$, defined at Line~11.
Hence due to submodular optimization, we have that
\[
\impr(S^*_\ell) \ge (1-1/e) \cdot \impr(F^*_\ell).
\]
Thus,
\begin{align*}
    \cost(C',S^*_\ell) &= \cost(C',F' \cup S^*_\ell) \\
    &= \cost(C',F') - \impr(S^*_\ell)\\
    &\le  \cost(C',F') - (1-1/e) \cdot \impr(F^*_\ell)\\
    &= \cost(C',F') - (1-1/e) \cdot \left( \cost(C',F') - \cost(C',F^*_\ell) \right) \\
    &= 1/e \cdot \cost(C',F') + (1-1/e) \cdot \cost(C',F^*_\ell)
\end{align*}
\end{proof}

The following proposition bounds $\cost(C',F')$ in terms of $\cost(C',F^*_\ell)$.
\begin{proposition} 
 $\cost(C',F') \le (3+2\eta) \cdot \cost(C',F^*_\ell)$.
\end{proposition}
Setting $\eta = \frac{e}{2} \epsilon'$, finishes the proof  for \divkmedian,
\begin{align*}
 \cost(C',S^*_\ell) &\le (3+ e \epsilon')/e \cdot \cost(C',F^*_\ell) + (1-1/e)
\cdot \cost(C',F^*_\ell)\\ &\le (1+2/e+\epsilon') \cdot \cost(C',F^*_\ell).
\end{align*}
For \divkmeans, setting $\eta = \frac{e}{16} \epsilon'$, finishes the proof of Proposition~\ref{lem:fptapxpm},
\begin{align*}
\cost(C',S^*_\ell) &\le (3+ 2e\epsilon'/16  )^2/e \cdot \cost(C',F^*_\ell) + (1-1/e)
 \cost(C',F^*_\ell)\\ &\le (1+8/e+\epsilon') \cdot \cost(C',F^*_\ell).
\end{align*}
for $\eta \le 1$.
%
\begin{proof}
To this end, it is sufficient to prove that for any client $c' \in C'$, it holds
that $d(c',F') \le (3+2\eta) d(c',F^*_\ell)$. Fix $c' \in C$, and let
$f^*_{\ell_j} \in F^*_\ell$ be the closest facility in $F^*_\ell$ with
$\lambda^*_j$ such that $\lambda^*_j = d(c'_{i_j},f^*_{\ell_j})$. Now,
\[
d(c',f^*_{\ell_j}) \ge d(c'_{i_j},f^*_{\ell_j}) \ge (1+\eta)^{([\lambda^*_j]_D -1)} \ge \frac{\lambda^*_j}{1+\eta}.
\]
Using, triangle inequality and the above equation, we have,
\begin{align*}
d(c',F') &\le d(c,F'_j) \le  d(c',F^*_{\ell_j}) + d(F^*_\ell,F') \le
d(c',f^*_{\ell_j}) + 2 \lambda^*_j \\&\le (3+2\eta)  d(c',f^*_{\ell_j}).
\end{align*}
\end{proof}

\subsection{Other proofs}\label{app:otherproofs}
\paragraph{Proof of Corollary~\ref{corollary:introduction:1}}
\seth implies that there is no $\bigO(|V|^{k-\epsilon})$ algorithm for
\dominatingset, for any $\epsilon>0$. The reduction
in~\cite[Lemma~1]{thejaswi2021diversity} creates an instance of \divkmedian (\divkmeans resp.)
where $F=V$. Hence, any \fpt exact or approximate algorithm running in time
$\bigO(|F|^{k-\epsilon})$ for \divkmedian (\divkmeans resp.) contradicts \seth.

$\hfill\square$
\paragraph{Proof of Proposition~\ref{proposition:introduction:2}}
First, note that any \fpt algorithm $\mathcal{A}$  achieving a multiplicative
factor approximation for \divkmedian (\divkmeans resp.) needs to solve the lower bound requirements
first. Since these requirements capture
\dominatingset~\cite[Lemma~1]{thejaswi2021diversity}, it means $\mathcal{A}$
solves \dominatingset in \fpt time, which is a contradiction. Finally, noting
the fact that finding even $f(k)$ size dominating set, for any $f(k)$, is also
\wone-hard due to~\cite{karthik2019on} finishes the proof.

$\hfill\square$

\end{document}